\documentclass[a4paper,10pt]{report}
\usepackage[utf8]{inputenc}

\usepackage{color}

\usepackage{amssymb}
\usepackage{amsmath}
\usepackage{amsfonts}
\usepackage{amsthm}
\usepackage{url}
\usepackage{makeidx}

\usepackage[dvips]{epsfig}
\usepackage[dvips]{graphics}

\usepackage{algorithm}
\usepackage{algorithmic}

\usepackage[all,knot,poly,arc]{xy}

\newcommand{\mi}{\mbox{$\mu$}}

\newcommand{\cD}{\mbox{$\cal D$}}

\newcommand{\cS}{\mbox{$\cal S$}}

\newcommand{\NN}{{\Bbb N}}
\newcommand{\RR}{{\Bbb R}}

\newcommand{\PP}{{\Bbb P}}

\newcommand{\bone}{\mbox{\boldmath$1$}}

\newcommand{{\uk}}{\mbox{$\underline{k}$}}

\def\nod(#1,#2){\put(#1,#2){\circle*{.125}}
\put(#1,#2){\makebox(0,0.5){{\small$#2$}}}}%
\def\rod(#1,#2){\put(#1,#2){\circle*{.2}}}
\def\NOD(#1,#2)#3{\put(#1,#2){\circle*{.2}}\put(#1,#2){\makebox(0,0.8){{\small$#3$}}}}

\def\EXX{{\hfill{$\diamondsuit$}}}

\makeindex

\newcounter{exampleNo}

\newtheorem{theorem}{Theorem}[chapter]
\newtheorem{lemma}[theorem]{Lemma}
\newtheorem{proposition}[theorem]{Proposition}
\newtheorem{corollary}[theorem]{Corollary}

\newenvironment{example}[1][Example \arabic{exampleNo}.]{\begin{trivlist}\refstepcounter{exampleNo}
\item[\hskip \labelsep {\bfseries #1}]}{\end{trivlist}}

\makeindex

\title{On the Theory of Stochastic Automata}
\author{Merve Nur Cakir, Mehwish Saleemi, Karl-Heinz Zimmermann\footnote{Email: k.zimmermann@tuhh.de}\\
Department of Computer Engineering \\
Hamburg University of Technology\\
21071 Hamburg, Germany}

\begin{document}
\maketitle
\tableofcontents

%
%
%

\chapter{Introduction}

The theory of discrete stochastic systems has been initiated by the work of 
Shannon~\cite{shannon} and von Neumann~\cite{neumann}.
While Shannon has considered memory-less communication channels and their generalization by introducing states,
von Neumann has studied the synthesis of reliable systems from unreliable components.
The fundamental work of Rabin and Scott~\cite{rscott} about deterministic finite-state automata 
has led to two generalizations.
First, the generalization of transition functions to conditional distributions 
studied by Carlyle~\cite{carl} and Starke~\cite{starke}.
This in turn has led to a generalization of discrete-time Markov chains in which the chains are governed 
by more than one transition probability matrix.
Second, the generalization of regular sets by introducing stochastic automata as described by Rabin~\cite{rabin}. 
Stochastic automata are well-investigated~\cite{claus}.

This report provides a short introduction to stochastic automata based on the valuable book of Claus~\cite{claus}.
This includes the basic topics of the theory of stochastic automata:
equivalence, minimalization, reduction, coverings, observability, and determinism.
Then stochastic versions of Mealy and Moore automata are studied and finally
stochastic language acceptors are considered as a generalization of nondeterministic finite-state acceptors.

\chapter{Stochastic Automata}
Stochastic automata are abstract machines with input and output behavior.
Such automata are also called transducers\index{transducer}.
A stochastic automaton can be viewed as an extension of a nondeterministic automaton with probabilistic transitions.

A {\em stochastic automaton\/}\index{stochastic automaton}  (SA) is a quadruple $A = (S,\Sigma, \Omega, p)$, where
\begin{itemize}
\item $S$ is a nonempty finite set of {\em states}\index{state},
\item $\Sigma$ is an alphabet of {\em input symbols}\index{input},
\item $\Omega$ is an alphabet of {\em output symbols}\index{output},  and
\item for each pair $(a,s)\in \Sigma\times S$, $p(\cdot,\cdot\mid a,s)$ 
is a conditional probability distribution on $\Omega\times S$.
\end{itemize}
Note that a conditional probability distribution $p(\cdot,\cdot\mid a,s)$ on $\Omega\times S$ consists of nonnegative numbers
$p(b,s'\mid a,s)$ for all $s'\in S$ and $b\in\Omega$ such that 
\begin{eqnarray}\label{e-SA-marg0}
\sum_{b\in \Omega} \sum_{s'\in S} p(b,s'\mid a,s) =1,\quad a\in\Sigma, s\in S.
\end{eqnarray}
Given a conditional probability distribution $p(\cdot,\cdot\mid a,s)$ on $\Omega\times S$,  
we define a conditional probability distribution $\hat p(\cdot,\cdot\mid x,s)$ with $x\in\Sigma^*$ and $s\in S$
on $\Omega^*\times S$ recursively as follows.
\begin{itemize}
\item For all $s,s'\in S$,
\begin{eqnarray}\label{e-SA-phat1}
\hat p (\epsilon,s'\mid \epsilon,s) = \left\{ \begin{array}{ll} 1 & \mbox{if } s=s',\\ 0 & \mbox{if } s\ne s',  \end{array} \right.
\end{eqnarray}
where $\epsilon$ denotes both the empty word in $\Sigma^*$ and $\Omega^*$.
\item For all $s,s'\in S$, $x\in\Sigma^*$, and $y\in\Omega^*$ with $|x|\ne |y|$,
\begin{eqnarray}\label{e-SA-phat2}
\hat p(y,s'\mid x,s) =  0.
\end{eqnarray}
\item For all $s,s'\in S$, $a\in \Sigma$, $x\in\Sigma^*$, $b\in\Omega$, and $y\in\Omega^*$,
\begin{eqnarray}\label{e-SA-phat3}
\hat p (yb,s'\mid xa,s) = \sum_{t\in S} \hat p(y,t\mid x,s)\cdot p(b,s'\mid a,t).
\end{eqnarray}
\end{itemize}
\begin{proposition}\label{p-SA-prob00}
For each pair $(x,s)\in\Sigma^*\times S$, $\hat p(\cdot,\cdot\mid x,s)$ is a conditional probability distribution 
on $\Omega^*\times S$.
\end{proposition}
Therefore,  for each pair $(x,s)\in\Sigma^*\times S$, 
\begin{eqnarray}\label{e-SA-phat4}
\sum_{y\in \Omega^*} \sum_{s'\in S} p(y,s'\mid x,s) =1.
\end{eqnarray}

A stochastic automaton works serially and synchronously.
It reads an input word symbol by symbol and after reading an input symbol 
it emits an output symbol and transits into another state.
More precisely, if the automaton starts in state $s$ and reads the word $x$, then
with probability $\hat p(y,s'\mid x,s)$ it will end in state $s'$ emitting the word $y$ by taking all intermediate states into account.

Note that the measures $p$ and $\hat p$ coincide on the set $\Omega\times S\times\Sigma\times S$ if we put $x=y=\epsilon$ in~(\ref{e-SA-phat3}).
Therefore, we will write $p$ instead of $\hat p$.
\begin{proposition}\label{p-SA-prob0}
For all $x,x'\in\Sigma^*$, $y,y'\in\Omega^*$, and $s,s'\in S$ with $|x|=|y|$, 
$$p(yy', s'\mid xx',s) = \sum_{t\in S} p(y, t\mid x,s) \cdot p(y', s'\mid x',t).$$ 
\end{proposition}
\begin{proof}
First, if $|x'|\ne |y'|$, then $|xx'|\ne |yy'|$ and so both sides are zero by~(\ref{e-SA-phat2}).
Second, if $x'=\epsilon=y'$, then both sides are equal to $p(y, s'\mid x,s)$ by~(\ref{e-SA-phat1}).

Third, if $x'\in \Sigma$ and $y'\in \Omega$, then both sides are equal by~(\ref{e-SA-phat3}).
Finally, suppose the assertion holds for words of length $|x'| = |y'|\leq k$ for some $k\geq 1$.
Then consider words $x'a$ and $y'b$ of length $k+1$, where $a\in\Sigma$ and $b\in\Omega$.
We have
\begin{eqnarray*}
\lefteqn{ p(yy'b, s'\mid xx'a,s) }\\
&=& \sum_{t\in S} p(yy', t\mid xx',s) \cdot p(b, s'\mid a,t),\quad \mbox{by~(\ref{e-SA-phat3})},\\
&=& \sum_{t\in S} \sum_{t'\in S} p(y, t'\mid x,s) p(y', t\mid x',t') \cdot p(b, s'\mid a,t),\;\mbox{by induction},\\
&=& \sum_{t'\in S} p(y, t'\mid x,s) \sum_{t\in S} p(y', t\mid x',t') \cdot p(b, s'\mid a,t)\\
&=& \sum_{t'\in S} p(y, t'\mid x,s) \cdot p(y'b, s'\mid x'a,t'),  \quad \mbox{by~(\ref{e-SA-phat3})}.
\end{eqnarray*}
\end{proof}


The behavior of a stochastic automaton can be described by probability matrices.
To see this, let $A$ be a stochastic automaton with state set $S=\{s_1,\ldots,s_n\}$.
For each pair of input and output symbols $a\in\Sigma$ and $b\in\Omega$, put
\begin{eqnarray}
p_{ij}(b\mid a) = p(b,s_j\mid a,s_i),\quad 1\leq i,j\leq n,
\end{eqnarray}
and define the real-valued $n\times n$ matrix 
\begin{eqnarray}
P(b\mid a) = (p_{ij}(b\mid a))_{1\leq i,j\leq n}.
\end{eqnarray}
Note that the matrix $P(b\mid a)$ is {\em substochastic}\index{matrix!substochastic}, 
i.e., it is a square matrix with nonnegative entries and by~(\ref{e-SA-phat4}) each row adds up to at most~1.
The elements of $P(b\mid a)$ provide the transition probabilities between the states 
if the symbol $a$ is read and the symbol $b$ is emitted.
This definition can be extended to strings of input and output symbols.
For this, note that by~(\ref{e-SA-phat1}) we have
\begin{eqnarray}
P(\epsilon\mid \epsilon) = I_n,
\end{eqnarray}
where $I_n$ is the $n\times n$ unit matrix.
Moreover, if $x\in\Sigma^*$ and $y\in\Omega^*$ with $|x|\ne |y|$, then by~(\ref{e-SA-phat2}) we have
\begin{eqnarray}
P(x\mid y) = O_n,
\end{eqnarray}
where $O_n$ is the $n\times n$ zero matrix.
Furthermore, if $a\in \Sigma$, $x\in\Sigma^*$, $b\in\Omega$, and $y\in\Omega^*$, then by~(\ref{e-SA-phat3}) we have
\begin{eqnarray}
P(yb\mid xa) = P(y\mid x)\cdot P(b\mid a).
\end{eqnarray}
By Prop.~\ref{p-SA-prob0} and the associativity of matrix multiplication, we obtain the following.
\begin{proposition}\label{p-SA-prob1}
For all $x,x'\in\Sigma^*$ and $y,y'\in\Omega^*$ with $|x|=|y|$, 
$$P(yy'\mid xx') = P(y\mid x) \cdot P(y'\mid x').$$ 
\end{proposition}
It follows by induction that if $x=x_1\ldots x_k\in\Sigma^*$ and $y=y_1\ldots y_k\in\Omega^*$, then
\begin{eqnarray}\label{e-SA-prob2}
P(y\mid x) = P(y_1\mid x_1)\cdots P(y_k\mid x_k).
\end{eqnarray}

\begin{proposition}
Each stochastic automaton $A$ is uniquely characterized by the collection of substochastic matrices
$$\{P(b\mid a) \mid a\in \Sigma, b\in \Omega\}.$$
For each input word $x\in\Sigma^*$, the matrix
$$P(x) = \sum_{y\in\Omega^*} P(y\mid x)$$
is (row) stochastic\index{matrix!stochastic}, i.e., each row sums up to~$1$.
\end{proposition}
\begin{proof}
The substochastic matrices $P(b\mid a)$ with $a\in \Sigma$ and $b\in \Omega$ describe the conditional probabilities of the automaton and so characterize the automaton uniquely.

By~(\ref{e-SA-marg0}), the matrix $P(a) = \sum_{b\in\Omega} P(b\mid a)$ is stochastic for each input symbol $a\in\Sigma$.
Moreover, the matrix $P(\epsilon)=P(\epsilon\mid\epsilon) = I_n$ is stochastic.
Since the multiplication of stochastic matrices is again a stochastic matrix, 
by~(\ref{e-SA-prob2}) the matrix $P(x)$ is stochastic for each $x\in\Sigma^*$.
\end{proof}

\begin{example}
Consider the stochastic automaton $A = (\{s_1,s_2\},\{a\},\{b\},p)$ with conditional probabilities 
$$p(b,s_1\mid a,s_1)=\frac{2}{3},\quad p(b,s_2\mid a,s_1)=\frac{1}{3},\quad\mbox{and}\quad p(b,s_2\mid a,s_2)=1.$$ 
The automaton is given by the state diagram in Fig.~\ref{fi-SA-A}.
The corresponding (substochastic) matrix is 
$$P(a) = P(b\mid a) = \left(\begin{array}{cc}\frac{2}{3}&\frac{1}{3}\\ 0 & 1\end{array}\right).$$
Thus for each integer $k \geq 1$,
$$P(a^k) = \left(\begin{array}{cc}\frac{2^k}{3^k}&\frac{3^k-2^k}{3^k}\\ 0 & 1\end{array}\right).$$
\EXX
\end{example}
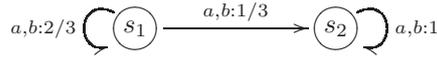
\begin{figure}[hbt]
\begin{center}
\mbox{$
\xymatrix{
*++[o][F-]{s_1} 
\ar@(ul,dl)[]_{a,b:2/3}/
\ar@{->}[rr]^{a,b:1/3}
&&
*++[o][F-]{s_2} 
\ar@(ur,dr)[]^{a,b:1}
}
$}
\end{center}
\caption{State diagram of $A$.}\label{fi-SA-A}
\end{figure}

For each input word $x\in\Sigma^*$, the stochastic matrix $P(x)$ can be viewed as generating a discrete-time Markov chain.
Thus the behavior of a stochastic automaton is an interleaving of Markov chains each of which corresponding to a single input symbol.

\begin{example}(Communication channels)\label{e-BSC}
A {\em binary symmetric channel\/}\index{binary symmetric channel} (BSC)\index{BSC} introduced by 
Claude Shannon (1948) is a common communication channel in coding theory.
A BSC with crossover probability $p$
is a binary input, binary output channel that flips the input bit with probability $p$.
A BSC is memoryless and characterized by the conditional probabilities 
\begin{eqnarray*}
\PP(Y=0\mid X=0) &=& 1-p, \\ \PP(Y=0\mid X=1) &=& p, \\ \PP(Y=1\mid X=0) &=& p, \\ \PP(Y=1\mid X=1) &=& 1-p,
\end{eqnarray*}
where the crossover probability lies between $0$ and $\frac{1}{2}$ (Fig.~\ref{f-BSC}).

An alternative to the more idealized binary symmetric channel is the {\em arbitrarily varying channel\/}\index{arbitrarily varying channel} (AVC)\index{AVC}
which is more realistic for network models.
An AVC has an input alphabet $\Sigma$, an output alphabet $\Omega$, and a state set $S$.
A symbol over the input alphabet is transmitted and at the other end a symbol over the output alphabet is received.
During transmission, the state in the set $S$ can vary arbitrarily at each time step.
The conditional probabilities of an AVC are defined as
$$p(b,s'\mid a,s) = p'(b\mid a,s)\cdot p''(s'\mid s).$$
where $p'(b\mid a,s)$ is the conditional probability of receiving the symbol $b$ when the symbol $a$ has been transmitted in state $s$,
and $p''(s'\mid s)$ is the conditional probability of moving from state $s$ to state $s'$.
\EXX
\end{example}
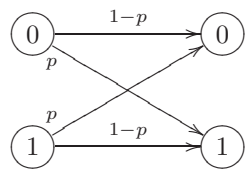
\begin{figure}[hbt]
\begin{center}
\mbox{$
\xymatrix{
*++[o][F-]{0} \ar@{->}[rr]^{1-p}\ar@{->}[drr]_<<{p}  && *++[o][F-]{0} \\
*++[o][F-]{1} \ar@{->}[rr]^{1-p}\ar@{->}[urr]^<<{p}  && *++[o][F-]{1} \\
}
$}
\end{center}
\caption{Binary symmetric channel.}\label{f-BSC}
\end{figure}


\chapter{Equivalence}
In the following, the capabilities of stochastic automata will be compared.
Intuitively, two stochastic automata have the same performance if they exhibit the same input-output behavior.

More concretely, let $A= (S,\Sigma,\Omega,p)$ be a stochastic automaton with state set $S=\{s_1,\ldots,s_n\}$.
The output behavior of the automaton being in state $s_i\in S$, reading the symbol $a\in \Sigma$ 
and emitting the symbol $y\in\Omega$ is given by the marginal probability 
\begin{eqnarray}
\eta_i(b\mid a) = \sum_{s\in S} p(b,s\mid a,s_i).
\end{eqnarray}
Take the column vector of marginals
\begin{eqnarray}
\eta(b\mid a) = \left( \begin{array}{c} \eta_1(b\mid a)\\\vdots\\\eta_n(b\mid a)\end{array} \right).
\end{eqnarray}
Thus $\eta(b\mid a)$ is the vector of row sums of the substochastic matrix $P(b\mid a)$.
Equivalently, if $\bone_n$ denotes the all-one vector of length $n$, we have
\begin{eqnarray}
\eta(b\mid a) = P(b\mid a)\cdot\bone_n.
\end{eqnarray}
This can be generalized to words $x\in\Sigma^*$ and $y\in\Omega^*$ by setting
\begin{eqnarray}\label{e-etaP}
\eta(y\mid x) = P(y\mid x)\cdot\bone_n.
\end{eqnarray}
The $i$th component $\eta_i(y\mid x)$ is the conditional probability of transition 
from state $s_i$ to any state when the word $x$ is read and the word $y$ is emitted, $1\leq i\leq n$.
The vector $\eta(y\mid x)$ is the {\em result vector\/}\index{result vector} of the stochastic automaton 
reading $x\in\Sigma^*$ and emitting $y\in \Omega^*$.
\begin{proposition}\label{p-omega}
We have
\begin{itemize}
\item
$\eta(\epsilon\mid\epsilon) = \bone_n.$
\item
For all words $x\in\Sigma^*$,
$$\sum_{y\in\Omega^*} \eta(y\mid x) =\bone_n.$$
\item
For all words $x,x'\in\Sigma^*$ and $y,y'\in\Omega^*$ with $|x|=|y|$, 
$$\eta(yy'\mid xx') =P(y\mid x)\cdot \eta(y'\mid x').$$
\end{itemize}
\end{proposition}
\begin{proof}
First, we have $p(\epsilon\mid\epsilon) = P(\epsilon\mid\epsilon)\cdot \bone_n=I_n\cdot\bone_n = \bone_n$.

Second, we have 
$$\sum_{y\in\Omega^*} \eta (y\mid x) = \sum_{y\in\Omega^*} P(y\mid x)\cdot\bone_n = P(x)\cdot \bone_n =\bone_n,$$
since the matrix $P(x)$ is stochastic.

Finally, by Prop.~\ref{p-SA-prob1} and~(\ref{e-etaP}), we have 
\begin{eqnarray*}
\eta (yy'\mid xx') &=& P(yy'\mid xx') \cdot\bone_n = P(y\mid x)\cdot P(y'\mid x') \cdot\bone_n \\
&=& P(y\mid x)\cdot \eta(y'\mid x').
\end{eqnarray*}
\end{proof}

A stochastic automaton can be thought of as an abstract machine which takes on a well-defined state at each step.  
However, an external observer is unable to determine this state with certainty.  
She can only consider the probability that the automaton is in a certain state.
Therefore, we introduce socalled state distributions.

Let $A= (S,\Sigma,\Omega,p)$ be a stochastic automaton with state set $S=\{s_1,\ldots,s_n\}$.
A row vector $\pi=(\pi_1,\ldots,\pi_n)$ is a {\em state distribution}\index{state distribution} of $A$ if the entries of $\pi$ are nonnegative and add up to~1.
The $i$th component $\pi_i$ is the probability that the automaton $A$ is in state $s_i$, $1\leq i\leq n$.

By convention, a state distribution $\pi=(\pi_1,\ldots,\pi_n)$ of $A$ will also be written 
as a formal linear combination of states,
\begin{eqnarray}
\pi = \pi_1s_1+\ldots+\pi_ns_n.
\end{eqnarray}
In particular, $\pi=s_i$ is the state distribution given by the $i$th unit vector $e_i$ which has~1 in position $i$ and 0's elsewhere, 
$1\leq i\leq n$.

Take the probability of emitting the word $y\in\Omega^*$ 
if the automaton $A$ is in state $s_i$ with probability $\pi_i$ and the word $x\in\Sigma^*$ is read.
This probability is given by the matrix product
\begin{eqnarray}
\pi\cdot \eta(y\mid x) = \sum_{i=1}^n \pi_i\cdot \eta_i(y\mid x).
\end{eqnarray}
For this, consider the mapping $ \eta^\pi :\Omega^*\times \Sigma^*\rightarrow[0,1]$ defined by
\begin{eqnarray}
 \eta^\pi (y\mid x) = \pi\cdot  \eta(y\mid x),\quad x\in\Sigma^*,y\in\Omega^*.
\end{eqnarray}
This mapping is well-defined, since $\sum_i\pi_i=1$, $0\leq  \eta_i(y\mid x)\leq 1$ for each $1\leq i\leq n$
and so $0\leq \sum_i \pi_i \eta_i(y\mid x)\leq 1$.

In view of the result vectors, define the sets
\begin{eqnarray}
\cD_A = \{ \eta^\pi\mid \pi\mbox{ state distribution of } A\}
\end{eqnarray}
and
\begin{eqnarray}
\cS_A = \{ \eta^s\mid s\mbox{ state of } \; A\}.
\end{eqnarray}

We introduce several notions of equivalence of state distributions and stochastic automata.
For this, let $A=(S_A,\Sigma,\Omega,p_A)$ and $B=(S_B,\Sigma,\Omega,p_B)$ be stochastic automata with the same input and output alphabets.
\begin{itemize}
\item
$A$ and $B$ are {\em equivalent\/}\index{stochastic automaton!equivalence}, written $A\equiv B$, if $\cD_A=\cD_B$.
\item
$A$ and $B$ are {\em S-equivalent\/}\index{stochastic automaton!S-equivalence}, written $A\equiv_S B$, if $\cS_A=\cS_B$.
\item
$A$ {\em covers} $B$\index{stochastic automaton!covering}, written $A\geq B$, if $\cD_A\supseteq\cD_B$.
\item
A state distribution $\pi_A$ of $A$ is {\em equivalent\/}\index{state distribution!equivalence} to a state distribution $\pi_B$ of $B$, 
written $\pi_A\equiv \pi_B$, if 
\begin{eqnarray}
\eta_A^{\pi_A}(y\mid x)  =  \eta_B^{\pi_B}(y\mid x)
\end{eqnarray}
for all $x\in\Sigma^*$ and $y\in\Omega^*$.
\item
Let $k\geq 0$.
A state distribution $\pi_A$ of $A$ is {\em $k$-equivalent\/}\index{state distribution!k-equivalence} to a state distribution $\pi_B$ of $B$, 
written $\pi_A\equiv_k \pi_B$, if 
\begin{eqnarray}
\eta_A^{\pi_A}(y\mid x)  =  \eta_B^{\pi_B}(y\mid x)
\end{eqnarray}
for all $x\in\Sigma^*$ and $y\in\Omega^*$ with $|x|=|y|\leq k$.
\end{itemize}
These relations are equivalence relations on the class of stochastic automata resp.\ the class of state distributions.
The following result follows directly from the definitions.
\begin{lemma}\label{l-SA-equiv}
Let $A$ and $B$ be stochastic automata with the same input and output alphabets.
Then $A\equiv B$ if and only if $A\geq B$ and $B\geq A$.
\end{lemma}

\begin{lemma}\label{l-SA-equiv2} 
Let $A$ and $B$ be stochastic automata.
Then $A$ covers $B$ if and only if for each state of $B$ there is an equivalent state distribution of $A$.
\end{lemma}
\begin{proof}
Let $S_A = \{s_1,\ldots,s_n\}$ and $S_B = \{t_1,\ldots,t_m\}$. 

First, suppose $A$ covers $B$.
Then for each state $t\in S_B$ there is a state distribution $\pi$ of $A$ such that
$\eta_B^t = \eta^\pi$.

Conversely, suppose that for each state $t\in S_B$ there is an equivalent state distribution $\pi(t)$ of $A$.
Then the state distribution $\pi' = \sum_j \pi'_jt_j $ of $B$ 
is equivalent to the state distribution $\pi = \sum_j \pi'_j\pi(t_j)$ of $A$.
Thus $\cD_A \supseteq \cD_B$ and hence $A$ covers $B$.
\end{proof}

\begin{proposition}\label{p-AB-equiv1}
Let $A$ and $B$ be stochastic automata with the same input and output alphabets.
If $A$ and $B$ are S-equivalent, then $A$ and $B$ are equivalent.
\end{proposition}
\begin{proof}
Let $A$ and $B$ be S-equivalent.
Then for each state $s_i$ in $A$ there is an equivalent state $t_i$ in $B$.
Thus the state distribution $\pi = \sum_j \pi_j s_j$ of $A$ is equivalent to the state distribution 
$\pi' = \sum_j \pi'_j t_j$ of $B$.
Thus $\cD_A\subseteq \cD_B$ and similarly $\cD_B\subseteq \cD_A$.
Hence, $A\equiv B$.
%
%
%
\end{proof}
This assertion gives rise to a decision algorithm for the equivalence of stochastic automata.
Instead of considering all state distributions, it is sufficient to take only the states (i.e., the state distributions of the states) into account (Alg.~\ref{a-SA-equiv}).
\begin{algorithm}\label{a-SA-equiv}
\caption{Equivalence of stochastic automata.}
\begin{algorithmic}
\REQUIRE Stochastic automata $A$ and $B$ with respective state sets $S_A=\{s_1,\ldots,s_m\}$ and $S_B=\{t_1,\ldots,t_n\}$ and common input and output alphabets
\ENSURE Output~1 if $A$ and $B$ are equivalent; otherwise, output~0
\STATE Compute $\cS_A = \{ \eta_A^{s_1},\ldots, p_A^{s_m}\}$
\STATE Compute $\cS_B = \{ \eta_B^{t_1},\ldots, p_B^{t_n}\}$
\IF  {$\cS_A = \cS_B$}
\RETURN 1  
\ELSE
\RETURN 0  
\ENDIF 
\end{algorithmic}
\end{algorithm}

\begin{example}
Consider the stochastic automaton $A=(\{s_1,s_2\},\{a,b\},\{c,d\},p)$ given by the substochastic matrices
$$
\begin{array}{ll}
P(c\mid a) = \left( \begin{array}{cc}\frac{1}{4}&\frac{1}{8}\\ 0&\frac{1}{2}\end{array} \right),\;&
P(d\mid a) = \left( \begin{array}{cc}\frac{1}{8}&\frac{1}{2}\\ 0&\frac{1}{2}\end{array} \right),\\
P(c\mid b) = \left( \begin{array}{cc}\frac{1}{4}&\frac{1}{4}\\ \frac{1}{4} & 0\end{array} \right),\;&
P(d\mid b) = \left( \begin{array}{cc}\frac{3}{8}&\frac{1}{8}\\ \frac{3}{4} & 0\end{array} \right).
\end{array}
$$
The corresponding stochastic matrices are
$$
P(a) = \left( \begin{array}{cc}\frac{3}{8}&\frac{5}{8}\\ 0&1\end{array} \right)\;\mbox{and}\;
P(b) = \left( \begin{array}{cc}\frac{5}{8}&\frac{3}{8}\\ 1&0\end{array} \right).
$$
The result vectors for single in- and outputs are
$$
\begin{array}{ll}
 \eta(c\mid a) = \left( \begin{array}{c}\frac{3}{8}\\ \frac{1}{2}\end{array} \right),\;&
 \eta(d\mid a) = \left( \begin{array}{c}\frac{5}{8}\\ \frac{1}{2}\end{array} \right),\\
 \eta(c\mid b) = \left( \begin{array}{c}\frac{1}{2}\\ \frac{1}{4}\end{array} \right),\; &
 \eta(d\mid b) = \left( \begin{array}{c}\frac{1}{2}\\ \frac{3}{4}\end{array} \right).
\end{array}
$$
For instance, the states $s_1$ and $s_2$ are not equivalent, since 
$ \eta^{s_1} (c\mid a) =\frac{3}{8}$ and $ \eta^{s_2} (c\mid a) =\frac{1}{2}$.
\EXX
\end{example}

%


\chapter{Reduction}
The objective of reduction is to construct stochastic automata with the same performance but a reduced number of states.
For this, a stochastic automaton $A$ 
is called {\em reduced\/}\index{stochastic automaton!reduced} if $A$ has no two distinct states which are equivalent.
\begin{proposition}\label{l-red}
Let $A$ and $B$ be reduced and S-equivalent stochastic automata.
Then the state sets of $A$ and $B$ have the same cardinality.
\end{proposition}
\begin{proof}
For each state $s\in S_A$ there is an equivalent state $s'\in S_B$ in the sense that the induced state distributions are equivalent.
Since the automaton $B$ is reduced, the state $s$ cannot be equivalent to another state $s''\in S_B$, since then the states $s'$ and $s''$ would be equivalent.
This gives an injective mapping $S_A\rightarrow S_B$.
In the same way, we obtain an injective mapping $S_B\rightarrow S_A$.
But the state sets are finite and hence there is a bijection between $S_A$ and $S_B$.
\end{proof}

For each stochastic automaton, there is an equivalent and reduced stochastic automaton 
which can be constructed using the powerset construction.
\begin{theorem}\label{t-redA}
Each stochastic automaton is S-equivalent to a reduced stochastic automaton.
\end{theorem}
\begin{proof}
Let $A=(S,\Sigma,\Omega,p_A)$ be a stochastic automaton with state set $S=\{s_1,\ldots,s_n\}$.
Two states $s$ and $t$ of $A$ are equivalent if $ \eta^s(y\mid x) =  \eta^t(y\mid x)$ for all $x\in\Sigma^*$ and $y\in\Omega^*$.
This relation is an equivalence relation on the state set $S$ of $A$.
Let $Z=\{Z_1,\ldots,Z_r\}$ denote the corresponding set of equivalence classes of the state set of $A$;
in particular, the set $Z$ is a partition of the state set $S$.
For each class $Z_i$ choose a fixed representative $z_i$, $1\leq i\leq r$.

Define the stochastic automaton $B=(Z,\Sigma,\Omega,p_B)$ with conditional probabilities 
\begin{eqnarray}\label{e-SA-AB}
p_B(b,Z_j\mid a,Z_i) = \sum_{z\in Z_j} p_A(b,z\mid a,z_i),\quad a\in\Sigma,\;b\in\Omega,\;1\leq i,j\leq r.
\end{eqnarray}
The automaton $B$ is reduced, since if the states $Z_i$ and $Z_j$ are equivalent, then $z_i=z_j$ and thus $i=j$ by definition.

Claim that the automata $A$ and $B$ are S-equivalent. 
Indeed, take the states $z_i\in Z_i$ in $A$ and $Z_i$ in $B$, $1\leq i\leq r$.
Then for all $a\in\Sigma$ and $b\in\Omega$,
\begin{eqnarray*}
 \eta_A^{z_i}(b\mid a) 
&=&  \sum_{s\in S_A} p_A(b,s\mid a,z_i) \\
&=& \sum_{Z_j\in Z}\sum_{z\in Z_j} p_A(b,z\mid a,z_i)\\
&=& \sum_{Z_j\in Z} p_B(b,Z_j\mid a,Z_i) \\
&=&   \eta_{B}^{Z_i}(b\mid a).
\end{eqnarray*}
Thus the states $z_i$ and $Z_i$ are 1-equivalent.
Suppose the states $z_i$ and $Z_i$ are $k$-equivalent for some $k\geq 1$.
Then for all $x\in\Sigma^*$, $y\in\Omega^*$ with $|x|=|y|=k$, $a\in\Sigma$, and $b\in\Omega$, we have
\begin{eqnarray*}
 \eta_A^{z_i}(by\mid ax) 
&=& \sum_{s\in S_A} p_A(b,s\mid a,z_i)  \eta_{A}^{s}(y\mid x), \quad \mbox{by Prop.~\ref{p-omega}},\\
&=& \sum_{Z_j\in Z} \sum_{s\in Z_j} p_A(b,s\mid a,z_i)  \eta_{A}^{s}(y\mid x) \\
&=& \sum_{Z_j\in Z} p_B(b,Z_j\mid a,Z_i)  \eta_{A}^{z_j}(y\mid x) \\
&=& \sum_{Z_j\in Z} p_B(b,Z_j\mid a,Z_i)  \eta_{B}^{Z_j}(y\mid x),\quad \mbox{by induction}, \\
&=&  \eta_{B}^{Z_i}(by\mid ax), \quad \mbox{by Prop.~\ref{p-omega}}.
\end{eqnarray*}
Therefore, the states $z_i$ and $Z_i$ are $(k+1)$-equivalent. 
Thus the states $z_i$ and $Z_i$ are equivalent, $1\leq i\leq r$, and hence the automata $A$ and $B$ are S-equivalent.
\end{proof}

The equivalence of two state distributions can be decided by considering inputs and outputs of restricted length.
\begin{proposition}\label{p-SA-s-equiv}
Let $A$ be a stochastic automaton with $n$ states.
Two state distributions of $A$ are equivalent if and only if they are $(n-1)$-equivalent.
\end{proposition}
\begin{proof}
Let $k\geq 0$.
Two state distributions $\pi$ and $\pi'$ of $A$ are $k$-equivalent if and only if
for all words $x\in \Sigma^*$ and $y\in\Omega^*$ with $|x|=|y|\leq k$,
$$\pi\cdot \eta(y\mid x) = \pi'\cdot \eta(y\mid x).$$
Equivalently,
$$(\pi-\pi')\cdot \eta(y\mid x) = 0.$$
Let $V_k$ be the vector subspace of $\RR^n$ generated by the result vectors
$ \eta(y\mid x)$, where $x\in \Sigma^*$ and $y\in\Omega^*$ with $|x|=|y|\leq k$.
We have $V_0 = \RR\cdot\bone_n$ and $V_k\subseteq V_{k+1}$ for each $k\geq 0$.
Thus $V_k$ is a subspace of $V_{k+1}$ for all $k\geq 0$.
Moreover, let $V_A$ be the vector subspace of $\RR^n$ generated by all the result vectors
$ \eta(y\mid x)$, where $x\in \Sigma^*$ and $y\in\Omega^*$.
Then $V_k$ is a subspace of $V_A$ for each $k\geq 0$.

Consider the ascending chain of subspaces of $\RR^n$:
$$V_0\subseteq V_1 \subseteq V_2\subseteq \ldots \subseteq V_A \subseteq \RR^n.$$ 
Since the subspaces are of finite dimension, there exists a smallest number $j\geq 0$ such that $V_j=V_{j+1}$. 
Then we have $V_j=V_{j+l}$ for each $l\geq 0$, 
since by Prop.~\ref{p-omega} the result vector $ \eta(ay\mid bx)$ is the image of the 
linear mapping $P(b\mid a)$ applied to the result vector $ \eta(y\mid x)$.
By construction, $V_j=V_A$ and thus by considering dimensions,
$$1=\dim V_0<\dim V_1< \dim V_2< \ldots <\dim V_j =\dim V_A \leq n=\dim \RR^n,$$ 
where $j\leq n-1$.

It follows that any two state distributions $\pi$ and $\pi'$ of $A$ are equivalent 
if and only if $(\pi-\pi')\cdot \eta(y\mid x) = 0$ for all $ \eta(y\mid x)\in V_{j}$.
Equivalently, the state distributions $\pi$ and $\pi'$ are $j$-equivalent with $j=n-1$.
\end{proof}

\begin{corollary}
Let $A$ and $B$ be stochastic automata with the same input and output alphabets as well as $m$ and $n$ states, respectively.
A state distribution of $A$ is equivalent to a state distribution of $B$ if and only if the state distributions are $(m+n-1)$-equivalent.
\end{corollary}
\begin{proof}
We may assume that the state sets of $A$ and $B$ are disjoint.
Define the stochastic automaton $C=(S,\Sigma,\Omega,p)$ with the same input and output alphabets, 
state set $S_C=S_A\cup S_B$, and substochastic matrices
$$P_C(b\mid a) = \left( \begin{array}{cc} P_A(b\mid a) & 0\\ 0 & P_B(b\mid a) \end{array} \right),\quad a\in\Sigma,\;b\in \Omega.$$
Then state distribution $\pi_A=(\pi_{A,1},\ldots,\pi_{A,m})$ of $A$ is equivalent to a state distribution $\pi_B=(\pi_{B,1},\ldots,\pi_{B,n})$ of $B$ if and only if
the length-$m+n$ state distributions 
$(\pi_{A,1},\ldots,\pi_{A,m},0\ldots,0)$ and $(0,\ldots,0,\pi_{B,1},\ldots,\pi_{B,n})$ of $C$ are equivalent.
By Prop.~\ref{p-SA-s-equiv}, the latter distributions are equivalent if and only if they are $(m+n-1)$-equivalent, since the automaton $C$ has $m+n$ states.
\end{proof}
The bound $n-1$ given in Prop.~\ref{p-SA-s-equiv} is tight as shown in the following.
\begin{example}
Let $n\geq 1$.
Consider the stochastic automaton 
$$A=(\{s_1,\ldots,s_n\},\{a,b\},\{c,d\},p)$$ with nonzero transition probabilities
$$p(c,s_{i+1}\mid a,s_i) = p(c,s_{i+1}\mid b,s_i) = 1,\quad 1\leq i\leq n-1,$$
and
$$p(c,s_1\mid a,s_n) = p(d,s_1\mid b,s_n) = 1.$$
Note that $A$ is a deterministic automaton (Fig.~\ref{fi-SA-An}).
If one starts in the state $s_1$ or $s_2$ and reads a word of length $l\leq n-2$, the word $c^l$ is emitted with certainty.
Thus the states $s_1$ and $s_2$ are $(n-2)$-equivalent.
However, if one begins in the state $s_1$ and inputs a word of length $n-1$, the word $c^{n-1}$ is emitted with certainty.
On the other hand, if one starts in the state $s_2$ and reads the word $b^{n-1}$, the word $c^{n-1}d$ is emitted with certainty.
Hence, the states $s_1$ and $s_2$ are not $(n-1)$-equivalent.
\EXX
\end{example}
\begin{figure}[hbt]
\begin{center}
\mbox{$
\xymatrix{
*++[o][F-]{s_1} \ar@{->}[rr]^{a,c:1;\; b,c:1} && *++[o][F-]{s_2} \ar@{->}[dd]^{a,c:1;\;b,c:1} \\
&&\\
*++[o][F-]{s_4} \ar@{->}[uu]^{a,c:1;\;b,d:1} && *++[o][F-]{s_3} \ar@{->}[ll]^{a,c:1;\;b,c:1} \\
}
$}
\end{center}
\caption{State diagram of stochastic automaton $A$ with $n=4$.}\label{fi-SA-An}
\end{figure}
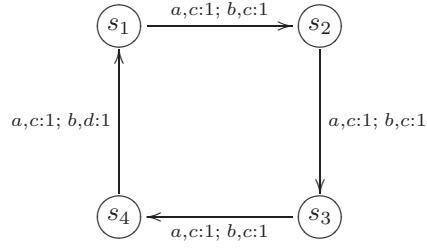

\begin{example}\label{e-sa-red}
Consider the stochastic automaton (Fig.~\ref{f-b34})
$$A=(\{s_1,s_2,s_3,s_4\},\{a\},\{b,c\},p)$$ with substochastic matrices
$$
P(b\mid a) =
\left(
\begin{array}{cccc}
0&0&\frac{1}{6}&\frac{1}{6}\\
0&0&\frac{1}{3}&\frac{1}{3}\\
\frac{1}{4}&\frac{1}{4}&0&0\\
0&0&0&\frac{1}{2}
\end{array}
\right)
\quad\mbox{and}\quad
P(c\mid a) =
\left(
\begin{array}{cccc}
0&0&\frac{1}{3}&\frac{1}{3}\\
0&0&\frac{1}{6}&\frac{1}{6}\\
\frac{1}{4}&\frac{1}{4}&0&0\\
0&0&0&\frac{1}{2}
\end{array}
\right).
$$
Then we have
$$
 \eta (b\mid a) =
\left(
\begin{array}{c}
\frac{1}{3}\\
\frac{2}{3}\\
\frac{1}{2}\\
\frac{1}{2}\\
\end{array}
\right)
\quad\mbox{and}\quad
 \eta (c\mid a) =
\left(
\begin{array}{c}
\frac{2}{3}\\
\frac{1}{3}\\
\frac{1}{2}\\
\frac{1}{2}\\
\end{array}
\right).
$$
By Prop.~\ref{p-omega}, we obtain 
$$
 \eta (bb\mid aa) =  \eta (bc\mid aa) =
\left(
\begin{array}{c}
\frac{1}{6}\\
\frac{1}{3}\\
\frac{1}{4}\\
\frac{1}{4}\\
\end{array}
\right)
=\frac{1}{2} \eta(b\mid a)
$$
and 
$$ \eta (cb\mid aa) =  \eta (cc\mid aa) =
\left(
\begin{array}{c}
\frac{1}{3}\\
\frac{1}{6}\\
\frac{1}{4}\\
\frac{1}{4}\\
\end{array}
\right)
=\frac{1}{2} \eta(c\mid a).
$$
Thus by the proof of Prop.~\ref{p-SA-s-equiv}, we have $V_0=\RR$ and $V_1=V_2=V_A$ with $\dim V_A=2$.
By Prop.~\ref{p-SA-s-equiv}, two states of $A$ are equivalent if and only if they are 1-equivalent.
The states $s_1$, $s_2$, and $s_3$ are pairwise non-equivalent, 
while the states $s_3$ and $s_4$ are 1-equivalent, since $ \eta^{s_3}(b\mid a) =  \eta^{s_4}(b\mid a)$ and $ \eta^{s_3}(c\mid a) =  \eta^{s_4}(c\mid a)$.
Thus the state set decomposed into equivalence classes gives $Z=\{Z_1=\{s_1\},Z_2=\{s_2\}, Z_3=\{s_3,s_4\}\}$.
\begin{figure}[hbt]
\begin{center}
\mbox{$
\xymatrix{
&& *++[o][F-]{s_1} \ar@/^/[lld]^{a,b:1/6;\; a,c:1/3}\ar@{->}[drr]^{a,b: 1/6;\; a,c:1/3} \\
*++[o][F-]{s_3} \ar@/^/[urr]^{a,b:1/4;\; a,c:1/4}\ar@/^/[drr]^{a,b:1/4;\; a,c:1/4} &&&& *++[o][F-]{s_4} \ar@(ur,dr)[]^{a,b:1/2;\;a:c: 1/2}\\
&& *++[o][F-]{s_2} \ar@/^/[ull]^{a,b:1/3;\;a,c:1/6}\ar@{->}[urr]_{a,b:1/3;\; a,c:1/6}\\
}
$}
\end{center}
\caption{State diagram of stochastic automaton $A$.}\label{f-b34}
\end{figure}

By~(\ref{e-SA-AB}), 
the corresponding reduced stochastic automaton $B$ (Fig.~\ref{f-b3}) 
with representative $s_3\in Z_3$ has the substochastic matrices
$$
P_B(b\mid a) =
\left(
\begin{array}{ccc}
0&0&\frac{1}{3}\\
0&0&\frac{2}{3}\\
\frac{1}{4}&\frac{1}{4}&0
\end{array}
\right)
\quad\mbox{and}\quad
P_B(c\mid a) =
\left(
\begin{array}{cccc}
0&0&\frac{2}{3}\\
0&0&\frac{1}{3}\\
\frac{1}{4}&\frac{1}{4}&0
\end{array}
\right),
$$
\begin{figure}[hbt]
\begin{center}
\mbox{$
\xymatrix{
&& *++[o][F-]{Z_1} \ar@/^/[lld]^{a,b:1/3;\; a,c:2/3} \\
*++[o][F-]{Z_3} \ar@/^/[urr]^{a,b:1/4;\; a,c:1/4}\ar@/^/[drr]^{a,b:1/4;\; a,c:1/4} \\
&& *++[o][F-]{Z_2} \ar@/^/[ull]^{a,b:2/3;\;a,c:1/3}\\
}
$}
\end{center}
\caption{State diagram of stochastic automaton $B$.}\label{f-b3}
\end{figure}

On the other hand, 
the associated reduced stochastic automaton $B'$ (Fig.~\ref{f-b4}) with representative $s_4\in Z_3$ has the substochastic matrices
$$
P_{B'}(b\mid a) =
\left(
\begin{array}{ccc}
0&0&\frac{1}{3}\\
0&0&\frac{2}{3}\\
0&0&\frac{1}{2}
\end{array}
\right)
\quad\mbox{and}\quad
P_{B'}(c\mid a) =
\left(
\begin{array}{cccc}
0&0&\frac{2}{3}\\
0&0&\frac{1}{3}\\
0&0&\frac{1}{2}
\end{array}
\right).
$$
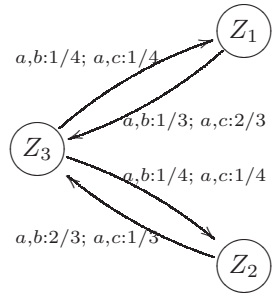
\begin{figure}[hbt]
\begin{center}
\mbox{$
\xymatrix{
&& *++[o][F-]{Z_1} \ar@{->}[lld]^{a,b:1/3;\; a,c:2/3} \\
*++[o][F-]{Z_3} \ar@(ul,dl)[]_{a,b:1/2;\; a,c:1/2}\\
&& *++[o][F-]{Z_2} \ar@{->}[ull]^{a,b:2/3;\;a,c:1/3}\\
}
$}
\end{center}
\caption{State diagram of stochastic automaton $B'$.}\label{f-b4}
\end{figure}
\EXX
\end{example}

Two stochastic automata $A=(S_A,\Sigma,\Omega,p_A)$ and $B=(S_B,\Sigma,\Omega,p_B)$ are
{\em isomorphic\/}\index{isomorphism} if there is a bijective mapping
$\phi:S_A\mapsto S_B$ such that
for all $s,s'\in S_A$, $a\in\Sigma$ and $b\in\Omega$,
\begin{eqnarray}
p_B(b,\phi(s')\mid a,\phi(s))  = p_A(b,s'\mid a,s).
\end{eqnarray}
\begin{example}
The reduced stochastic automata $B$ and $B'$ in Ex.~\ref{e-sa-red} are not isomorphic.
The only mapping that would make sense is 
$\phi:\{s_1,s_2,s_3\} \rightarrow\{s_1,s_2,s_4\}$ with 
$\phi(s_1)=s_1$, $\phi(s_2)=s_2$, and $\phi(s_3)=s_4$.
\EXX
\end{example}

The vector space $V_A$ associated with a stochastic automaton $A$ 
can be used to further investigate the problem of equivalence of state distributions.
For this, we need a vector-space basis of the space $V_A$.
To this end, let $\Sigma$ and $\Omega$ be alphabets.
The set 
\begin{eqnarray}\label{e-orderM}
M=\{(y\mid x)\mid x\in\Sigma^*,y\in\Omega^*,|x|=|y|\} 
\end{eqnarray}
together with the binary operation of concatenation, 
\begin{eqnarray}
(y\mid x)(y'\mid x') = (yy'\mid xx'), \quad x,x'\in\Sigma^*,\; y,y'\in\Omega^*,
\end{eqnarray}
is a free monoid generated by the set $N = \{(b\mid a)\mid a\in\Sigma,b\in\Omega\}$.
Note that the monoids $M$ and $(\Omega\times \Sigma)^*$ are isomorphic by the assignment
\begin{eqnarray}
(b_1\ldots b_n\mid a_1\ldots a_n)\mapsto (b_1,a_1)\ldots(b_n,a_n).
\end{eqnarray}
This isomorphism maps the generating set $N$ to the cartesian product $\Omega\times \Sigma$. 
In this way, we may assume that the elements of $N$ are lexicographically ordered.
This ordering can be extended to the monoid $M$.

Let $A$ be a stochastic automaton with $n$ states.
Consider the vector subspace $V_A$ of $\RR^n$ generated by the result vectors
$ \eta(y\mid x)$, where $x\in \Sigma^*$ and $y\in\Omega^*$.
Suppose the space $V_A$ has $\RR$-dimension $d$.
Then define the $n\times d$ matrix 
\begin{eqnarray}
H_A = \left(\begin{array}{ccc} \eta(y_1\mid x_1)&\ldots&  \eta(y_d\mid x_d)\end{array}\right),
\end{eqnarray}
where $x_1=y_1=\epsilon$ and
$(y_{i+1}\mid x_{i+1})\in M$ is the lexicographically smallest element of $M$ that
is linearly independent of the previous vectors $ \eta(y_1\mid x_1),\ldots,  \eta(y_i\mid x_i)$, $1\leq i\leq d-1$.
The columns of the matrix $H_A$ form a basis of the space $V_A$.
Since $d\leq n$, the matrix $H_A$ has rank $d$.

\begin{example}
Reconsider the stochastic automaton $A$ in Ex.~\ref{e-sa-red}.
The matrix $H_A$ has rank~2 and is given by
$$H_A
= \left(  \eta(\epsilon,\epsilon)\;   \eta(b\mid a)\right) 
= \left(
\begin{array}{cc}
1 & \frac{1}{3}\\ 
1 & \frac{2}{3}\\ 
1 & \frac{1}{2}\\ 
1 & \frac{1}{2} 
\end{array}
\right).$$
\EXX
\end{example}

The matrix $H_A$ can be used to characterize equivalent state distributions.
\begin{proposition}\label{p-HA-pi}
Let $A$ be a stochastic automaton.
Two state distributions $\pi$ and $\pi'$ of $A$ are equivalent if and only if 
$\pi H_A=\pi' H_A$.

In particular, two states $s_i$ and~$s_j$ of $A$ are equivalent if and only if 
the $i$th and $j$th rows of $H_A$ are equal.
\end{proposition}
\begin{proof}
Let $\pi$ and $\pi'$ be equivalent state distributions.
Then $\pi\cdot \eta(y\mid x) = \pi'\cdot \eta(y\mid x)$ for all $x\in\Sigma^*$ and $y\in\Omega^*$.
This particularly holds for the columns of the matrix $H_A$ and therefore $\pi H_A=\pi' H_A$.

Conversely, let $\pi H_A=\pi' H_A$ for state distributions $\pi$ and $\pi'$.
The columns of the matrix $H_A$ form a basis of the space $V_A$ and so we have for each result vector $ \eta(y\mid x)\in V_A$, 
$$ \eta(y\mid x) = \sum_{i=1}^d \alpha_i \eta(y_i\mid x_i),\quad \alpha_i\in\RR, \;1\leq i\leq d.$$
By hypothesis, we have $\pi\cdot \eta(y_i\mid x_i) = \pi'\cdot \eta(y_i\mid x_i)$ for each $1\leq i\leq d$.
Thus $\pi \cdot \eta(y\mid x) = \pi'\cdot \eta(y\mid x)$.
Hence, $\pi$ and $\pi'$ are equivalent.

The second assertion is an immediate consequence of the first one.
\end{proof}

\begin{proposition}\label{p-H-pi}
Let $A$ and $B$ be reduced stochastic automata.
If $A$ and $B$ are S-equivalent, then $H_A = H_B$ up to a permutation of matrix rows.
\end{proposition}
\begin{proof}
Suppose $A$ and $B$ have the same state set.
We may assume that the state $s$ in $A$ is equivalent to the state $s$ in $B$. 
Then $\eta_A(y\mid x) = \eta_B(y\mid x)$ for all $x\in\Sigma^*$ and $y\in\Omega^*$ by the ordering of these vectors.
Therefore, $H_A=H_B$.

Suppose $A$ and $B$ have the different state sets.
Then by Prop.~\ref{l-red}, the matrices can only differ by a row permutation.
\end{proof}

\begin{theorem}\label{th-4}
Let $A=(S,\Sigma,\Omega,p)$ be a stochastic automaton with $n$ states given by the collection of matrices 
$\{P_A(b\mid a)\mid a\in\Sigma,b\in\Omega\}$.

Let $\{P_B(b\mid a)\mid a\in\Sigma,b\in\Omega\}$ be a collection of non-negative $n\times n$ matrices 
such that
$$P_A(b\mid a) H_A = P_B(b\mid a)H_A$$
for all $a\in\Sigma$ and $b\in\Omega$. 
Then the collection
$\{P_B(b\mid a)\mid a\in\Sigma,b\in\Omega\}$ 
defines a stochastic automaton $B$ with $n$ states such that the state $s\in S_A$ is equivalent to the state
$s\in S_B$.
In particular, the automata $A$ and $B$ are S-equivalent.
\end{theorem}
\begin{proof}
Since $\bone_n$ is the first column of $H_A$, we have $P_A(b\mid a) \cdot\bone_n = P_B(b\mid a)\cdot\bone_n$.
Then 
$$\sum_b P_B(b\mid a)\cdot\bone_n = \sum_b P_A(b\mid a)\cdot\bone_n = P_A(a)\cdot\bone_n = \bone_n.$$ 
Thus the collection of matrices $\{P_B(b\mid a)\mid a\in\Sigma,b\in\Omega\}$ 
defines a stochastic automaton $B$ with $P_B(a) = \sum_b P_B(b\mid a)$.

Let $s\in S_A$.
Since for all $a\in\Sigma$ and $b\in\Omega$,
$$\eta_A(b\mid a) = P_A(b\mid a)\cdot\bone_n = P_B(b\mid a)\cdot\bone_n = \eta_B(b\mid a),$$
it follows that $s$ in $A$ is 1-equivalent to $s$ in $B$.

Suppose for some $k\geq 1$, $s$ in $A$ is $k$-equivalent to $s$ in $B$.
Then for all $x\in\Sigma^k$, $a\in\Sigma$, $y\in \Omega^k$ and $b\in\Omega$,
\begin{eqnarray*}
\eta_A(by\mid ax) 
&=& P_A(b\mid a) \eta_A(y\mid x),\quad \mbox{by Prop.~\ref{p-omega}},\\
&=& P_B(b\mid a) \eta_A(y\mid x) \\
&=& P_B(b\mid a) \eta_B(y\mid x),\quad \mbox{by induction},\\
&=& \eta_B(by\mid ax),\quad  \mbox{by Prop.~\ref{p-omega}},
\end{eqnarray*}
where the second equation uses the fact that $\eta_A(y\mid x)$ is a linear combination of the columns of $H_A$.
Thus $s$ in $A$ is $k+1$-equivalent to $s$ in $B$.
Hence, $A$ and $B$ are S-equivalent.
\end{proof}

The following result provides the reverse statement. 
\begin{theorem}
Let $A$ and $B$ be stochastic automata with the same state set such that
state $s$ in $A$ is equivalent to state $s$ in $B$ and
$B$ is defined by the collection of matrices
$\{P_B(b\mid a)\mid a\in\Sigma,b\in\Omega\}$. 
Then for all $a\in\Sigma$ and $b\in\Omega$,
$$P_A(b\mid a) H_A = P_B(b\mid a)H_A.$$
\end{theorem}
\begin{proof}
Since state $s$ in $A$ is equivalent to state $s$ in $B$, we have 
$\eta_A(y\mid x) = \eta_B(y\mid x)$ for all $x\in\Sigma^*$ and $y\in\Omega^*$.
Then by Prop.~\ref{p-omega},
$$P_A(b\mid a)\eta_A(y\mid x)  
= \eta_A(by\mid ax)  
= \eta_B(by\mid ax)  
= P_B(b\mid a)\eta_B(y\mid x) .$$
Since all columns of $H_A$ are of the form $\eta_A(y\mid x)$, it follows that
$P_A(b\mid a) H_A = P_B(b\mid a)H_A$ for all $a\in\Sigma$ and $b\in\Omega$.
\end{proof}

\begin{corollary}
Let $A$ be a stochastic automaton with $n$ states.
If the vector space $V_A$ has dimension $d=n$, then $A$ is reduced. 
\end{corollary}
\begin{proof}
Let $d=n$.
Then the $n\times d$ matrix $H_A$ which has independent columns has full rank and therefore is invertible.
However, by Prop.~\ref{p-HA-pi}, two states $s_i$ and~$s_j$ are equivalent if and only if 
the $i$th and $j$th rows of $H_A$ are equal.
Therefore, $A$ is reduced.
\end{proof}
The converse of this result is not valid.
\begin{example}
Take the stochastic automaton $A = (\{s_1,s_2,s_3\},\{a\},\{b,c\},p)$, where
$$ P(b\mid a) = \left( \begin{array}{ccc}
1 & 0 & 0\\
0 & 0 & 0\\
\frac{1}{2} & 0 & 0\\ 
\end{array} \right)
\quad\mbox{and}\quad
P(c\mid a) = \left( \begin{array}{ccc}
0 & 0 & 0\\
0 & 1 & 0\\
0 & \frac{1}{2} & 0\\ 
\end{array} \right).
$$
Then 
$$ H_A = \left( \begin{array}{cc}
1 & 1 \\
1 & 0 \\
1 & \frac{1}{2} \\
\end{array} \right).
$$
By Prop.~\ref{p-HA-pi}, the automaton $A$ is reduced, but the matrix $H_A$ does not have full rank.
\EXX
\end{example}

\chapter{Minimality}
The objective is to further decrease the state set of a reduced stochastic automaton.
For this, a stochastic automaton $A$ is called {\em minimal\/}\index{stochastic automaton!minimal} if no state of $A$ is equivalent to another state distribution of $A$.
That is, if $S_A=\{s_1,\ldots,s_n\}$, 
there is no state $s_i$ and no state distribution $\pi$ different from $s_i$ such that $ \eta^{s_i}= \eta^\pi$.

\begin{example}\label{e-sa-red-min}
Reconsider the reduced stochastic automaton $B'$ in Ex.~\ref{e-sa-red}.
This automaton is not minimal.
To see this, note that 
$$H_{B'} 
= \left( \eta(\epsilon\mid\epsilon)\; \eta(b\mid a)\right)
= \left( \begin{array}{cc}
1 & \frac{1}{3}\\
1 & \frac{2}{3}\\
1 & \frac{1}{2}
\end{array} \right).$$
For the state distributions $e_3 = (0,0,1)$ and $\pi=(\frac{1}{2},\frac{1}{2},0)$,
we have $e_3 \cdot H_A=\pi\cdot H_A$.
Thus by Prop.~\ref{p-HA-pi}, the state distributions $s_3$ and $\pi$ are equivalent.
\EXX
\end{example}

A {\em permutation matrix\/}\index{permutation matrix} is a quadratic matrix with entries 0 and~1 which has exactly one entry~1 in each row and each column.
An $n\times n$ permutation matrix represents a bijection of a set of $n$ elements.

\begin{lemma}\label{l-perm}
Let $M$ be a stochastic $m\times n$ matrix and $N$ be a stochastic $n\times m$ matrix.
If $M\cdot N=I_m$ and $N\cdot M=I_n$, then $m=n$ and the matrices $M$ and $N$ are permutation matrices.
\end{lemma}
\begin{proof}
Matrices can be viewed as linear mappings between vector spaces.
If $M\cdot N=I_m$, the linear mappings given by $M$ and $N$ are surjective and injective, respectively.
Similarly, if $N\cdot M=I_n$, the linear mappings given by $M$ and $N$ are injective and surjective, respectively.
Therefore, the linear mappings determined by $M$ and $N$ are bijective; i.e., the matrices $M$ and $N$ are invertible with $m=n$.

Let $M=(a_{ij})$ and $N=(b_{ij})$ be $n\times n$ matrices with $M\cdot N=I_n$ and $N\cdot M=I_n$.
Then for the $(i,i)$th diagonal element, $1\leq i\leq n$, we have $\sum_{j=1}^n a_{ij}b_{ji} = 1$.
Since $M$ and $N$ are nonnegative matrices, we have $a_{il}>0$ and $b_{li}>0$ for some $l$.
Moreover, for each $j\ne l$, we have $\sum_{k=1}^n b_{jk}a_{kl} = 0$.
Since $a_{il}\ne 0$, we have $b_{ji}=0$. 
Thus $1 = \sum_{j=1}^n a_{ij}b_{ji} = a_{il}b_{lj}$.
But $a_{il},b_{lj}\leq 1$ and so $a_{il}=b_{lj}=1$.
Therefore, each row of $M$ and each column of $N$ has exactly one entry~$1$ and the remaining entries are~$0$.
By exchanging the roles of $M$ and $N$,
each row of $N$ and each column of $M$ has exactly one entry~$1$ and the remaining entries are~$0$.
Hence, $M$ and $N$ are permutation matrices.
\end{proof}

\begin{proposition}\label{p-AB-equiv2}
Let $A$ and $B$ be minimal stochastic automata with the same input and output alphabets.
Then $A$ and $B$ are equivalent if and only if $A$ and $B$ are S-equivalent.
\end{proposition}
\begin{proof}
By Prop.~\ref{p-AB-equiv1}, $A$ and $B$ are equivalent if they are S-equivalent.

Conversely, let $A$ and $B$ be equivalent.
Let $S_A=\{s_1,\ldots,s_m\}$ and $S_B=\{t_1,\ldots,t_n\}$ be the state sets of $A$ and $B$, respectively.
Then for each state $s_i\in S_A$ there is an equivalent state distribution $\pi'_i$ of $B$, 
and for each state $t_j\in S_B$  there is an equivalent state distribution $\pi_j$ of $A$.
Write
$$
\pi'_i = \sum_{k=1}^n c_{ik}t_k
\quad\mbox{and}\quad 
\pi_j = \sum_{l=1}^m d_{jl}s_l.
$$
Thus the state $s_i$ is equivalent to the state distribution
$$
\sum_{k=1}^n \sum_{l=1}^m c_{ik}d_{kl}s_l =  \sum_{l=1}^m s_l \sum_{k=1}^n c_{ik}d_{kl}.
$$
Since the automaton $A$ is minimal, the right-hand of the equation must be equal to $s_i$;
that is,
$$
\sum_{k=1}^n c_{ik}d_{kl} = \left\{ \begin{array}{ll} 1 & \mbox{if } i=l,\\ 0 & \mbox{if } i\ne l.  \end{array} \right.
$$
Take the $n\times m$ matrix $C=(c_{ij})$ and the $m\times n$ matrix $D=(d_{jk})$.
Both matrices are stochastic, since their rows are state distributions.
Moreover, the above equation shows that $C\cdot D=I_n$ and $D\cdot C=I_m$.
Therefore, by Lemma~\ref{l-perm}, we have $m=n$ and the matrices $C$ and $D$ are permutation matrices.
In particular, the state distributions $\pi_j$ and $\pi'_i$ are unit vectors 
and therefore correspond to the states of $A$ and $B$, respectively.
Thus for each state in $A$ there is an equivalent state in $S_B$, and vice versa.
Hence, the automata $A$ and $B$ are S-equivalent.
\end{proof}

The construction of minimal stochastic automata is sketched in the next result.
\begin{theorem}\label{t-sa-min}
Each reduced stochastic automaton $A$ is equivalent to a minimal stochastic automaton.
\end{theorem}
\begin{proof}
Consider the set $Z_A$ of all states $s\in S_A$ for which there is a state distribution $\pi$ of $A$ such that 
$s$ and $\pi$ are distinct and equivalent.
After rearranging the state set, write $S_A=\{s_1,\ldots,s_n\}$ and $Z_A=\{s_1,\ldots,s_k\}$ with $k\leq n$. 
If $Z_A$ is the empty set, the automaton $A$ is already minimal.

Suppose $Z_A\ne \emptyset$.
Consider the state $s_1\in Z_A$.
There is a state distribution $\pi'_1$ of $A$ such that $s_1$ and $\pi'_1$ are distinct and equivalent.
Write $\pi'_1 = \sum_{i=1}^n\pi'_{1i}s_i$.
Since $s_1$ and $\pi'_1$ are different, we have $\pi'_{11}\ne 1$.
Then by Prop.~\ref{p-HA-pi}, the state $s_1$ is equivalent to the state distribution 
$$\pi_1=\sum_{i=2}^n \frac{\pi'_{1i}}{1-\pi'_{11}} s_i,$$
where $\pi_{11}=0$ and $\pi_i = \frac{\pi'_{1i}}{1-\pi'_{11}}$ for each $2\leq i\leq n$.

More generally, one can show by induction that for each $1\leq l \leq k$, the following holds:
For each state $s_i$, $1\leq i\leq l$, there is a state distribution $\pi_i$ of $A$ 
which is equivalent to $s_i$ and has the property that $\pi_{i1}=\ldots=\pi_{il}=0$.
In particular, the case $l=k$ shows that for each state $s_i$, $1\leq i\leq k$, 
there is a state distribution $\pi_i$ of $A$ which is equivalent to $s_i$ and has the property that 
$\pi_{i1}=\ldots=\pi_{ik}=0$.

Take the stochastic automaton $B$ with state set $S_B=S_A\setminus Z_A = \{s_{k+1},\ldots,s_n\}$ 
and conditional probabilities
$$p_B(b,s_j\mid a,s) = p_A(b,s_j\mid a,s) + \sum_{i=1}^k\pi_{ij}\cdot p_A(b,s_i\mid a,s)$$
for all $a\in\Sigma$, $b\in\Omega$, and $s,s_j\in S_B$.
By definition, the automaton $B$ is minimal.
Claim that the automata $A$ and $B$ are equivalent.
Indeed, define the $k\times (n-k)$ matrix
$$\Pi = \left( \begin{array}{ccc}
\pi_{1,k+1} & \ldots & \pi_{1,n}\\
\vdots      &  \ddots& \vdots\\
\pi_{k,k+1} & \ldots & \pi_{k,n}\\
\end{array} \right),$$
the $(n-k)\times n$ matrix 
$$R = \left( O_{n-k,k} \; I_{n-k}\right)$$
where $O_{n-k,k}$ denotes $(n-k)\times k$ zero matrix and $I_{n-k}$ the $(n-k)\times (n-k)$ unit matrix,
and the $n\times (n-k)$ matrix 
$$S = \left( \begin{array}{c}
\Pi \\ I_{n-k}
\end{array} \right).$$
This gives the $n\times n$ matrix
$$SR = \left( \begin{array}{cc}
O_{k,k} & \Pi \\ O_{n-k,k} & I_{n-k}
\end{array} \right).$$
Since the state $s_i$ is equivalent to the state distribution $\pi_i$ for $1\leq i\leq k$, we obtain
\begin{eqnarray}\label{e-SR}
SR \cdot \eta_A(y\mid x) =  \eta_A(y\mid x),\quad x\in\Sigma^*,\;y\in\Omega^*.
\end{eqnarray}
Moreover, the substochastic matrices defining $A$ and $B$ satisfy
\begin{eqnarray}\label{e-RS}
R \cdot P_A(b\mid a)\cdot S = P_B(b\mid a),\quad a\in\Sigma,\;b\in\Omega.
\end{eqnarray}
Claim that for all $x\in\Sigma^*$ and $y\in\Omega^*$ with $|x|=|y|$, 
$$ \eta_B(y\mid x) = R \cdot  \eta_A(y\mid x).$$
Indeed, the equation holds for $x=y=\epsilon$.
Suppose the equation is valid for all $x\in\Sigma^*$ and $y\in\Omega^*$ with $|x|=|y|\leq l$ for some $l\geq 0$. 
Then we obtain for $a\in\Sigma$, $b\in\Omega$, $x\in\Sigma^*$, and $y\in\Omega^*$ with $|x|=|y|=l$,
\begin{eqnarray*}
 \eta_B(by\mid ax) 
&=& P_B(b\mid a)\cdot \eta_B(y\mid x),\quad\mbox{by Prop.~\ref{p-omega}},\\
&=& R \cdot P_A(b\mid a)\cdot S \cdot \eta_B(y\mid x),\quad \mbox{by~(\ref{e-RS})},\\
&=& R \cdot P_A(b\mid a)\cdot SR \cdot  \eta_A(y\mid x), \quad \mbox{by induction},\\
&=& R \cdot P_A(b\mid a)\cdot  \eta_A(y\mid x),\quad \mbox{by~(\ref{e-SR})},\\
&=& R \cdot  \eta_A(by\mid ax),\quad \mbox{by Prop.~\ref{p-omega}}.
\end{eqnarray*}
By definition of the matrix $R$,
the vector $\eta_B(y\mid x) = R \cdot  \eta_A(y\mid x)$ consists of the last $n-k$ components of the vector $ \eta_A(y\mid x)$.
It follows that state $s_i$ of $A$ is equivalent to the state $s_i$ of $B$ for each $k+1\leq i\leq n$.
Moreover, for each $1\leq i\leq k$, the state $s_i$ of $A$ is equivalent to the state distribution $\pi_i$ of $B$.
Thus for each state of $A$ there is an equivalent state distribution of $B$ and 
for each state of $B$ there is an equivalent state distribution of $A$.
Hence, by the proof of Prop.~\ref{p-AB-equiv1}, the automata $A$ and $B$ are equivalent.
\end{proof}

Note that the minimal stochastic automaton $B$ can be successively constructed from the given stochastic automaton $A$.
Suppose the state $s_1\in S_A$ is equivalent to the state distribution $\pi=(0,\pi_2,\ldots,\pi_n)$.
Note that this distribution can be obtained as in the above proof.
Then take the stochastic automaton $B_1$ with state set $S_{B_1} = S_A\setminus\{s_1\}$ and conditional probabilities
\begin{eqnarray}
p_{B_1}(b,s_j\mid a,s) = p_A(b,s_j\mid a,s) + \pi_j\cdot p_A(b,s_1\mid a,s)
\end{eqnarray}
for all $a\in\Sigma$, $b\in\Omega$, and $s,s_j\in S_{B_1}$.
The iteration of this process leads to a minimal stochastic automaton equivalent to $A$. 

\begin{example}
Take the stochastic automaton $A = (\{s_1,s_2,s_3\}, \{a\},\{b,c\},p)$ with substochastic matrices 
$$
P_A(b\mid a) =
\left(
\begin{array}{ccc}
0&0&\frac{1}{2}\\
0&0&\frac{1}{3}\\
0&0&\frac{2}{3}
\end{array}
\right)
\quad\mbox{and}\quad
P_A(c\mid a) =
\left(
\begin{array}{ccc}
0&0&\frac{1}{2}\\
0&0&\frac{2}{3}\\
0&0&\frac{1}{3}
\end{array}
\right).
$$
This gives the matrix
$$ H_A = 
\begin{pmatrix} 
1& \frac{1}{2}\\
1& \frac{1}{3}\\
1& \frac{2}{3}
\end{pmatrix}.$$
This automaton is reduced by Prop.~\ref{p-HA-pi} 
but not minimal, since the state $s_1$ and the state distribution $\pi =(0,\frac{1}{2},\frac{1}{2})$ are equivalent.
Therefore, define the stochastic automaton $B_1$ with state set $S_{B_1} =\{s_2,s_3\}$ and substochastic matrices
$$
P_{B_1}(b\mid a) =
\left(
\begin{array}{cc}
0&\frac{1}{3}\\
0&\frac{2}{3}
\end{array}
\right)
\quad\mbox{and}\quad
P_{B_1}(c\mid a) =
\left(
\begin{array}{cc}
0&\frac{2}{3}\\
0&\frac{1}{3}
\end{array}
\right).
$$
This gives the matrix
$$ H_{B_1} = 
\begin{pmatrix} 
1& \frac{1}{3}\\
1& \frac{2}{3}
\end{pmatrix}.$$
This automaton is already minimal.
\EXX
\end{example}

A vector $u\in\RR^n$ is a {\em convex combination\/}\index{convex combination} of vectors $v_1,\ldots,v_k\in\RR^n$ 
if there are nonnegative real numbers $\lambda_1,\ldots,\lambda_k$ such that
$$\sum_{i=1}^k \lambda_i v_i= u\quad\mbox{and}\quad \sum_{i=1}^k \lambda_i = 1.$$
The minimality of a stochastic automaton can be tested by using the $H$-matrices.
\begin{proposition}\label{p-sa-min-HA}
A stochastic automaton $A$ is minimal if and only if no row of $H_A$ is a convex combination of the other rows.
\end{proposition}
\begin{proof}
Suppose the automaton $A$ is not minimal.
Then there is a state $s_i$ and a state distribution $\pi$ such that $s_i$ and $\pi$ are distinct and equivalent.
As in the proof of Thm.~\ref{t-sa-min}, the state distribution $\pi$ can be chosen such that the $i$th component is $\pi_i=0$.
By Prop.~\ref{p-HA-pi}, we have $s_i H_A = \pi H_A$.
This equation says that the $i$th row of $H_A$ given by $s_i H_A$ 
is a convex combination of the remaining rows of $H_A$ given by $\pi H_A$.
\end{proof}

\begin{example}(Shimon~Even, 1965)
Let $\Sigma=\{a,b\}$, $\Omega=\{c,d,e\}$, and $S=\{s_1,\ldots,s_5\}$.
Consider the stochastic automaton $A = (S,\Sigma,\Omega,p_A)$ with substochastic matrices
$$
\begin{array}{ll}
P_A(c\mid a) = \left( \begin{array}{ccccc}
0 & \frac{1}{2} & 0 & 0 & 0 \\
0 & \frac{1}{2} & 0 & 0 & 0 \\
0 & 0 & 0 & 0 & 0 \\
0 & 0 & 0 & 0 & 0 \\
0 & 0 & 0 & 0 & 0 
\end{array} \right), &	
P_A(d\mid a) = \left( \begin{array}{ccccc}
0 & 0 & 0 & 0 & 0 \\
0 & 0 & 0 & 0 & 0 \\
0 & 0 & 0 & \frac{1}{2} & 0 \\
0 & 0 & 0 & \frac{1}{2} & 0 \\
0 & 0 & 0 & 0 & 0 
\end{array} \right),\\
P_A(e\mid a) = \left( \begin{array}{ccccc}
0 & 0 & 0 & 0 & \frac{1}{2}  \\
0 & 0 & 0 & 0 & \frac{1}{2}  \\
0 & 0 & 0 & 0 & \frac{1}{2}  \\
0 & 0 & 0 & 0 & \frac{1}{2}  \\
\frac{1}{2}  & 0  & \frac{1}{2}  & 0 & 0
\end{array} \right), & 
P_A(c\mid b) = \left( \begin{array}{ccccc}
\frac{1}{2} & 0 & 0 & 0 & 0 \\
0 & 0 & 0 & 0 & 0 \\
0 & 0 & 0 & 0 & 0 \\
\frac{1}{2} & 0 & 0 & 0 & 0 \\
0 & 0 & 0 & 0 & 0 \\
\end{array} \right),\\
P_A(d\mid b) = \left( \begin{array}{ccccc}
0 & 0 & 0 & 0 & 0 \\
0 & 0 & \frac{1}{2} & 0 & 0 \\
0 & 0 & \frac{1}{2} & 0 & 0 \\
0 & 0 & 0 & 0 & 0 \\
0 & 0 & 0 & 0 & 0 \\
\end{array} \right), &
P_A(e\mid b) = \left( \begin{array}{ccccc}
0 & 0 & 0 & 0 & \frac{1}{2}  \\
0 & 0 & 0 & 0 & \frac{1}{2}  \\
0 & 0 & 0 & 0 & \frac{1}{2}  \\
0 & 0 & 0 & 0 & \frac{1}{2}  \\
\frac{1}{2}  & 0  & \frac{1}{2}  & 0 & 0
\end{array} 
\right),
\end{array} 
$$
and the stochastic automaton $B = (S,\Sigma,\Omega,p_B)$ with substochastic matrices
$P_B(c\mid a) =P_A(c\mid a)$,
$P_B(d\mid a) =P_A(d\mid a)$,
$P_B(c\mid b) =P_A(c\mid b)$,
$P_B(d\mid b) =P_A(d\mid b)$, and
$$P_B(e\mid a) = \left( \begin{array}{ccccc}
0 & 0 & 0 & 0 & \frac{1}{2} \\
0 & 0 & 0 & 0 & \frac{1}{2} \\
0 & 0 & 0 & 0 & \frac{1}{2} \\
0 & 0 & 0 & 0 & \frac{1}{2} \\
\frac{1}{4} & \frac{1}{4} & \frac{1}{4} & \frac{1}{4} & 0
\end{array} \right) = P_B(e\mid b).
$$
The vector space $V_A$ has the basis elements
$ \eta_A(\epsilon\mid \epsilon)$,
$ \eta_A(c\mid a)$, 
$ \eta_A(c\mid b)$, 
and
$ \eta_A(e\mid a)$;
that is, 
$$H_A = \left( \begin{array}{cccc}
1 & \frac{1}{2} & \frac{1}{2} & \frac{1}{2} \\ 
1 & \frac{1}{2} & 0 & \frac{1}{2} \\
1 & 0 & 0 & \frac{1}{2} \\
1 & 0 & \frac{1}{2} & \frac{1}{2} \\
1 & 0 & 0 & 1
\end{array} \right).$$
Furthermore, we have $H_A=H_B$.
It is easy to check that no row of $H_A$ is a convex combination of the other rows.
Thus both automata are minimal by Prop.~\ref{p-sa-min-HA} and both automata are S-equivalent by~Thm.~\ref{th-4}.
Hence, the automata are equivalent by Prop.~\ref{p-AB-equiv2}.
However, the automata $A$ and $B$ are not isomorphic.
\EXX
\end{example}
This example shows that there are stochastic automata which are minimal and equivalent, but not isomorphic.

%

Finally, the aim is to further reduce the state set of a minimal stochastic automaton. 
For this, a stochastic automaton $A$ is called {\em strongly reduced\/}\index{stochastic automaton!strongly reduced} 
if equivalent state distributions of $A$ are equal.
By definition, each strongly reduced stochastic automaton is minimal.  
However, not every minimal stochastic automaton is strongly reduced.
\begin{example}
In the example of Even, the state distributions 
$\pi=(\frac{1}{2},0,\frac{1}{2},0,0)$ and $\pi'=(0,\frac{1}{2},0,\frac{1}{2},0)$ are equivalent, since
$\pi H_A = (1, \frac{1}{4}, \frac{1}{4}, \frac{1}{2}) =  \pi' H_A$.
It follows that the minimal automaton $A$ is not strongly reduced.
\EXX
\end{example}

Strong reduction of a stochastic automaton can be tested by using the $H$-matrices.
\begin{proposition}\label{p-strongr}
Let $A$ be a stochastic automaton with $n$ states.
Then $A$ is strongly reduced if and only if the matrix $H_A$ has rank $n$.
\end{proposition}
\begin{proof}
Suppose the $n\times d$ matrix $H_A$ has rank $d=n$.
Then $H_A$ is an $n\times n$ matrix which determines a bijective linear mapping.
Let $\pi$ and $\pi'$ be equivalent state distributions of $A$.
Then by Prop.~\ref{p-HA-pi}, we have $\pi H_A=\pi' H_A$; i.e., $(\pi-\pi')H_A=0$.
But the kernel of the bijective linear mapping provided by $H_A$ is the zero space $\{0\}$.
Thus we have $\pi=\pi'$ and hence the automaton $A$ is strongly reduced.

Conversely, suppose the $n\times d$ matrix $H_A$ has rank $d<n$.
Then the kernel of the linear mapping given by $H_A$ has nonzero dimension.
That is, there is a nonzero vector $u\in\RR^n$ such that $u H_A=0$.
Since the first column of the matrix $H_A$ is the all-one vector, 
we have $\sum_{i=1}^nu_i=0$.
Moreover, since the vector $u$ is nonzero and the matrix $H_A$ is nonnegative, 
the vector $u$ must have nonzero positive and nonzero negative components.
Let $u^+$ be the vector, which results from $u$ by setting all negative entries to~0.
Moreover, put $u^- = u^+-u$; i.e., $u^-$ results from $u$ by setting all positive entries to~0 and then changing signs. 
Then the vectors $u^+$ and $u^-$ are nonnegative with $u = u^+-u^-$.
For instance, if $u=(1,-4,3,0)$, then $u^+=(1,0,3,0)$ and $u^-=(0,4,0,0)$.
Since $\sum_{i=1}^nu_i=0$, we have $\alpha = \sum_i u_i^+  = \sum_i u_i^-$ 
and so $\pi^+=\frac{1}{\alpha}u^+$ and $\pi^-=\frac{1}{\alpha}u^-$ are state distributions of $A$ with 
$(\pi^+-\pi^-)H_A=0$; i.e., $\pi^+ H_A=\pi^- H_A$.
Thus by Prop.~\ref{p-HA-pi}, the distributions $\pi^+$ and $\pi^-$ are equivalent and 
hence the automaton $A$ is not strongly reduced.
\end{proof}

\begin{example}
Each reduced stochastic automaton $A$ with two states is strongly reduced.
To see this, consider the corresponding matrix
$$H_A = \begin{pmatrix} 1& a_1\\ 1& a_2 \end{pmatrix}.$$
Suppose the state distributions $(b,1-b)$ and $(c,1-c)$ are equivalent.
Then equivalently $(b,1-b) H_A = (c,1-c)H_A$ and so 
$a_1(b-c) = a_2(b-c)$.
If the distributions are distinct, then $b\ne c$ and thus $a_1=a_2$.
It follows that the matrix has rank $1$ and so $A$ is not reduced contradicting the hypothesis.
\EXX
\end{example}

A geometric interpretation can help to tackle the question 
whether a stochastic automaton $A$ is reduced, minimal or strongly reduced.
These properties only depend on the matrix $H_A$ and so its obvious to further analyse the vectors of this matrix.
For this, let $H_A$ be an $n\times d$ matrix with row vectors $h_1,\ldots,h_n$.
These vectors span an $(n-1)$-dimensional simplex\index{simplex}
\begin{eqnarray}
C_A = \left\{ \sum_{i=1}^n c_ih_i\mid \sum_{i=1}^n c_i = 1\mbox{ and } c_i\geq 0 \mbox{ for }1\leq i\leq n\right\}.
\end{eqnarray}
Note that an $(n-1)$-simplex is an $(n-1)$-dimensional polytope which is given by the convex hull of its $n$ vertices.

\begin{proposition}\label{p-simplex}
Let $A$ be a stochastic automaton.
\begin{itemize}
\item $A$ is not reduced if and only if two generating vectors of\/ $C_A$ are equal.
\item $A$ is not minimal if and only if one of the generating vectors of\/ $C_A$ is not a vertex.
\item $A$ is not strongly reduced if and only if\/ $C_A$ has dimension less than $n-1$.
\end{itemize}
\end{proposition}
The proof follows from the Props.~\ref{p-HA-pi},~\ref{p-sa-min-HA}, and~\ref{p-strongr}.

\begin{example}
The stochastic automaton $A$ in Fig.~\ref{f-b34} has the matrix
$$H_A = \begin{pmatrix}
1 & \frac{1}{3}\\
1 & \frac{2}{3}\\
1 & \frac{1}{2}\\
1 & \frac{1}{2}\\
\end{pmatrix}$$
and thus is not reduced.
The stochastic automata $B$ in Fig.~\ref{f-b3} has the matrix
$$H_B = \begin{pmatrix}
1 & \frac{1}{3}\\
1 & \frac{2}{3}\\
1 & \frac{1}{2}\\
\end{pmatrix}.$$
The vectors
$h_1= (1, \frac{1}{3})$,
$h_2= (1, \frac{2}{3})$,
and
$h_3= (1, \frac{1}{2})$ span a line segment or $1$-simplex with vertices $h_1$ and $h_2$, while $h_3$ lies in-between.
Therefore, $B$ is not minimal.
\EXX
\end{example}

\begin{example}
The matrix $H_A$ of the stochastic automaton $A$ in Even's example is spanned by the vectors
\begin{eqnarray*}
h_1 &=& ( 1, \frac{1}{2}, \frac{1}{2}, \frac{1}{2}), \\ 
h_2 &=& ( 1, \frac{1}{2}, 0, \frac{1}{2}), \\
h_3 &=& ( 1, 0, 0, \frac{1}{2}), \\
h_4 &=& ( 1, 0, \frac{1}{2}, \frac{1}{2}), \\
h_5 &=& ( 1, 0, 0, 1).
\end{eqnarray*}
These vectors span a pyramid with a quadratic base, where the apex is perpendicular above one of the vertices.
Therefore, the automaton $A$ is minimal.
Since the vectors span a $3$-dimensional simplex, the automaton is not strongly reduced.
\EXX
\end{example}


\chapter{Coverings}

The previous chapters have shown that the state set of a stochastic automaton can eventually be decreased 
without diminishing its performance.
Further reductions are possible if coverings of stochastic automata are considered.

For this, let $A=(S_A,\Sigma,\Omega,p_A)$ and $B=(S_B,\Sigma,\Omega,p_B)$ be stochastic automata, and
let $S_A$ and $S_B$ have $n$ and $m$ states, respectively.
The automaton $B$ is a {\em stochastic image\/}\index{stochastic image} 
of automaton $A$, written $A\rightarrow_s B$, if 
there is a stochastic $m\times n$ matrix $Q$ such that for all $x\in\Sigma^*$ and $y\in\Omega^*$,
\begin{eqnarray}
\eta_B(y\mid x) = Q \cdot \eta_A(y\mid x).
\end{eqnarray}
Note that the matrix $Q$ defines a linear mapping $\phi:\RR^n\rightarrow\RR^m:v\mapsto Qv$ such that $\phi$ maps the result vectors
$\eta_A(y\mid x)$ to the result vectors $\eta_B(y\mid x)$.

\begin{proposition}\label{p-sh-c}
Let $A$ and $B$ be stochastic automata with the same input and output alphabets.
Then $A$ covers $B$ if and only if $B$ is a stochastic image of $A$. 
\end{proposition}
\begin{proof}
Let $S_A = \{s_1,\ldots,s_n\}$ and $S_B = \{t_1,\ldots,t_m\}$.

Suppose $A$ covers $B$.
Then by Lemma~\ref{l-SA-equiv2}, for each state $t\in S_B$ there is an equivalent state distribution $q_j$ of $A$.
Let $Q$ be the matrix with rows $q_1,\ldots,q_m$.
Then $\eta_{B}^{t_j}(y\mid x) = q_j\cdot\eta_A(y\mid x)$ for all $x\in\Sigma^*$, $y\in\Omega^*$ and $1\leq j\leq m$.
Thus
$\eta_{B}(y\mid x) = Q\cdot\eta_A(y\mid x)$ for all $x\in\Sigma^*$, $y\in\Omega^*$.
Hence, $B$ is a stochastic image of $A$.

Conversely, let $B$ be a stochastic image of $A$.
The $j$th row $q_i$ of the matrix $Q$ is a state distribution of $A$.
Then $\eta_{B}^{t_j}(y\mid x) = q_j\cdot\eta_A(y\mid x)$ for all $x\in\Sigma^*$, $y\in\Omega^*$ and $1\leq j\leq m$.
Thus the state $t_j$ is equivalent to the state distribution $q_j$ of $A$.
By Lemma~\ref{l-SA-equiv2}, it follows that $A$ covers $B$.
\end{proof}

Stochastic homomorphy is reflected by the corresponding $H$-matrices.
\begin{proposition}
Let $A$ and $B$ be stochastic automata with the same input and output alphabets.
If there is an $m\times n$ matrix $Q$ such that $B$ is a stochastic image of $A$, 
then $H_B$ is obtained from $Q\cdot H_A$ by deleting linearly dependent columns.
If the columns of $Q\cdot H_A$ are linearly independent, then $H_B = Q\cdot H_A$.
\end{proposition}
\begin{proof}
Let $\eta_1,\ldots,\eta_d$ be the columns of the $n\times d$ matrix $H_A$.
Then $\eta'_i = Q\cdot\eta_i$ for $1\leq i\leq d$ are the columns of the $m\times d$ matrix $Q\cdot H_A$.
Since each result vector $\eta_A(y\mid x)$ of $A$ is a linear combination of the vectors $\eta_1,\ldots,\eta_d$,
the vector $\eta_B(y\mid x)  =  Q\cdot\eta_A(y\mid x)$ is a linear combination of the vectors
$\eta'_1=Q\cdot\eta_1,\ldots,\eta'_d=Q\cdot\eta_d$.

Put $C = QH_A$.  
We may assume that the columns of $C$ are totally ordered as given by the set $M$ in~(\ref{e-orderM}).
It is easy to check that the matrix $D$ obtained from $C$ by deleting in turn the columns 
which are linear combinations of previous columns is exactly the matrix $H_B$.
Moreover, if the columns of $QH_A$ are linearly independent, then $C=D=QH_A$.
\end{proof}

\begin{example}
The stochastic automata $A$ and $B'$ in Figs.~\ref{f-b34} and~\ref{f-b4} have the respective matrices
$$
H_A = 
\begin{pmatrix}
1 & \frac{1}{3}\\
1 & \frac{2}{3}\\
1 & \frac{1}{2}\\
1 & \frac{1}{2}\\
\end{pmatrix}
\quad\mbox{and}\quad
H_{B'} = 
\begin{pmatrix}
1 & \frac{1}{3}\\
1 & \frac{2}{3}\\
1 & \frac{1}{2}\\
\end{pmatrix}.
$$
Then $B'$ is a stochastic image of $A$ by using the matrix $Q = \begin{pmatrix}  & 0\\ I_3 & 0 \\ &0\end{pmatrix}$.
\EXX
\end{example}

\begin{proposition}
Let $A$ and $B$ be stochastic automata with the same input and output alphabets.
Suppose $A$ and $B$ have $n$ and $m$ states, respectively. 
Then $B$ is a stochastic image of $A$ if and only if
there exists a stochastic $m\times n$ matrix~$Q$ such that for all $a\in\Sigma$ and $b\in\Omega$,
$$Q\cdot P_A(b\mid a)\cdot H_A = P_B(b\mid a) \cdot Q\cdot H_A.$$ 
\end{proposition}
\begin{proof}
Let $B$ be a stochastic image of $A$.
That is, there is a stochastic $m\times n$ matrix $Q$ such that
$Q\cdot \eta_A(y\mid x) = \eta_B(y\mid x)$ for all $x\in\Sigma^*$ and $y\in\Omega^*$. 

Equivalently,
$Q\cdot \eta_A(y\mid x) = \eta_B(y\mid x)$ for all $x\in\Sigma^+$ and $y\in\Omega^+$, 
since $Q\cdot \eta_A(\epsilon\mid\epsilon) = \eta_B(\epsilon\mid \epsilon)$. 

Equivalently,
$Q\cdot \eta_A(by\mid ax) =  \eta_B(by\mid ax)$ 
for all $a\in\Sigma$, $b\in \Omega$, $x\in\Sigma^*$ and $y\in\Omega^*$. 

Equivalently, by Prop.~\ref{p-omega},
$Q\cdot P_A(b\mid a) \cdot \eta_A(y\mid x) = P_B(b\mid a) \cdot \eta_B(y\mid x)$ 
for all $a\in\Sigma$, $b\in \Omega$, $x\in\Sigma^*$ and $y\in\Omega^*$. 

Equivalently, 
$Q\cdot P_A(b\mid a) \cdot \eta_A(y\mid x) = P_B(b\mid a) \cdot Q\cdot \eta_A(y\mid x)$ 
for all $a\in\Sigma$, $b\in \Omega$, $x\in\Sigma^*$ and $y\in\Omega^*$. 

Equivalently,
$Q\cdot P_A(b\mid a) \cdot H_A  = P_B(b\mid a) \cdot Q \cdot H_A$ for all $a\in\Sigma$ and $b\in\Omega$, since
the columns of $H_A$ form a basis of the set of result vectors in $A$. 
\end{proof}

General homomorphisms between stochastic automata are introduced.
For this, let $A = (S_A,\Sigma,\Omega,p_A)$ and $B = (S_B,\Sigma,\Omega,p_B)$ be stochastic automata.
An {\em S-homomorphism}\index{S-homomorphism} from $A$ to $B$ is given by a mapping $\phi:S_A\rightarrow S_B$
such that for all $a\in\Sigma$, $b\in\Omega$ and $s,s'\in S_A$,
\begin{eqnarray} 
p_B(b,\phi(s')\mid a,\phi(s)) 
&=& p_A(b,\phi^{-1}\phi(s')\mid a,s)\\
&=& \sum_{s''\in S_A\atop \phi(s'')=\phi(s')} p_A(b,s''\mid a,s),\nonumber
\end{eqnarray}
where $\phi^{-1}\phi(s')= \{s''\in S_A\mid\phi(s'')=\phi(s')\}$. 

An S-homomorphism from $A$ to $B$ is also written as $\phi:A\rightarrow B$.
In particular, an {\em S-homomorphism} $\phi:A\rightarrow B$
is called {\em S-epimorphism} if the mapping $\phi$ is surjective.
Each stochastic automaton $A$ can be mapped by an S-epimomorphism 
to the {\em trivial stochastic automaton\/}\index{stochastic automaton!trivial} $B=(\{a\},\{b\},\{s_0\},p)$
with $p(b,s_0\mid a,s_0)= 1$.

\begin{proposition}\label{p-Shomxy}
Let $A = (S_A,\Sigma,\Omega,p_A)$ and $B = (S_B,\Sigma,\Omega,p_B)$ be stochastic automata.
Each S-homomorphism $\phi:A\rightarrow B$ satisfies
\begin{eqnarray} 
p_B(y,\phi(s')\mid x,\phi(s)) = p_A(y,\phi^{-1}\phi(s')\mid x,s),
\end{eqnarray}
where $x\in\Sigma^*$, $y\in\Omega^*$ and $s,s'\in S_A$.
\end{proposition}
This assertion can be proved by induction on the length of the input words.

\begin{proposition}\label{p-S-hom}
Let $\phi:A\rightarrow B$ be an S-homomorphism.
Then each state $s$ in $A$ is equivalent to the state $\phi(s)$ in $B$.
In particular, if $\phi$ is an S-epimorphism, then $A$ and $B$ are S-equivalent.
\end{proposition}
\begin{proof}
For all $x\in\Sigma^*$, $y\in\Omega^*$, and $s\in S_A$, we have by Prop.~\ref{p-Shomxy}, 
$$\eta_B^{\phi(s)}(y\mid x) = \eta_A^s(y\mid x).$$
The results follow.
\end{proof}

\begin{example}
Consider the stochastic automaton (Fig.~\ref{f-f24})
$$A=(\{a\},\{b,c\},\{s_1,s_2,s_3,s_4\},p_A),$$ 
where
$$
P_A(b\mid a) = \left( \begin{array}{cccc}
0 & 0 & \frac{1}{6} & \frac{1}{6} \\ 
0 & 0 &\frac{2}{3} & 0 \\
0 & 0 & 0 & \frac{1}{2} \\
0 & 0 & 0 & \frac{1}{2} \\
\end{array} \right)
\quad\mbox{and}\quad
P_A(c\mid a) = \left( \begin{array}{cccc}
0 & 0 & \frac{1}{3} & \frac{1}{3} \\ 
0 & 0 &\frac{1}{3} & 0 \\
0 & 0 &  \frac{1}{2} & 0\\
0 & 0 & 0 & \frac{1}{2} \\
\end{array} \right)
$$
and the stochastic automaton (Fig.~\ref{f-b40})
$$B=(\{a\},\{b,c\},\{t_1,t_2,t_3\},p_B),$$ 
where
$$
P_B(b\mid a) = \left( \begin{array}{ccc}
0 & 0 & \frac{1}{3} \\ 
0 & 0 &\frac{2}{3} \\
0 & 0 & \frac{1}{2} \\
\end{array} \right)
\quad\mbox{and}\quad
P_B(c\mid a) = \left( \begin{array}{ccc}
0 & 0 & \frac{2}{3} \\ 
0 & 0 &\frac{1}{3}  \\
0 & 0 &  \frac{1}{2} \\
\end{array} \right).
$$
The mapping $\phi:S_A\rightarrow S_B$ given by $\phi(s_1)=t_1$, $\phi(s_2)=t_2$ and $\phi(s_3)=\phi(s_4)=t_3$ is
an S-epimorphism.
For instance,
$p_B(b,t_3\mid a,t_1) = p_A(b,s_3\mid a,s_1) + p_A(b,s_4\mid a,s_1) = \frac{1}{6} + \frac{1}{6} = \frac{1}{3}$
and
$p_B(b,t_3\mid a,t_2) = p_A(b,s_3\mid a,s_2) + p_A(b,s_4\mid a,s_2) = \frac{2}{3} $.
\EXX
\end{example}
\begin{figure}[hbt]
\begin{center}
\mbox{$
\xymatrix{
&& *++[o][F-]{s_1} \ar@{->}[lld]_{a,b:1/6;\; a,c:1/3} \ar@{->}[rrd]^{a,b:1/6;\; a,c:1/3} \\
*++[o][F-]{s_3} \ar@(ul,dl)[]_{a,c:1/2}\ar@{->}[rrrr]^{a,b:\;1/2} &&&& *++[o][F-]{s_3} \ar@(ur,dr)[]^{a,b:1/2;\; a,c:1/2}\\
&& *++[o][F-]{s_2} \ar@{->}[ull]^{a,b:2/3;\;a,c:1/3}\\
}
$}
\end{center}
\caption{State diagram of stochastic automaton $A$.}\label{f-f24}
\end{figure}
\begin{figure}[hbt]
\begin{center}
\mbox{$
\xymatrix{
&& *++[o][F-]{t_1} \ar@{->}[lld]^{a,b:1/3;\; a,c:2/3} \\
*++[o][F-]{t_3} \ar@(ul,dl)[]_{a,b:1/2;\; a,c:1/2}\\
&& *++[o][F-]{t_2} \ar@{->}[ull]^{a,b:2/3;\;a,c:1/3}\\
}
$}
\end{center}
\caption{State diagram of stochastic automaton $B$.}\label{f-b40}
\end{figure}

\chapter{Observability and Determinism}

This chapter considers stochastic automata with restricted transition probability distributions.

A stochastic automaton $A=(S,\Sigma,\Omega,p)$ is called {\em observable\/}\index{stochastic automaton!observable}
if there is a (partial) mapping $\gamma: \Omega\times \Sigma\times S \rightarrow S$ such that for all
$a\in\Sigma$, $b\in\Omega$, and $s,s'\in S$,
\begin{eqnarray}
p(b,s'\mid a,s)\ne 0\quad\Rightarrow\quad s' = \gamma(b,a,s).
\end{eqnarray}
If in an observable stochastic automaton the current state, the input and the output are known, 
the next state is uniquely determined.

The (partial) mapping $\gamma: \Omega\times \Sigma\times S \rightarrow S$ can be extended 
to the (partial) mapping $\gamma: \Omega^*\times \Sigma^*\times S \rightarrow S$ such that
for all $a\in\Sigma$, $x\in\Sigma^*$, $b\in\Omega$, $y\in\Omega^*$, $s,s'\in S$ with $|x|=|y|$, 
\begin{eqnarray}
\gamma(\epsilon,\epsilon,s) &=& s,\label{e-gamma-p0}\\
\gamma(by,ax,s) &=& \gamma(y,x,\gamma(b,a,s)). \label{e-gamma-p}
\end{eqnarray}

\begin{lemma}\label{l-Ao}
A stochastic automaton $A$ is observable if and only if
for all $a\in\Sigma$ and $b\in\Omega$, 
the matrices $P(b\mid a)$ contain in each row at most one nonzero entry.
\end{lemma}
\begin{proof}
Let $A$ be observable and let $S=\{s_1,\ldots,s_n\}$.
Consider the $s_i$th row of the matrix $P(b\mid a)$.
Since $p(b,s_j\mid a,s_i)\ne 0$ implies $s_j=\gamma(b,a,s_i)$, the $s_i$th row contains 
$p(b,s_j\mid a,s_i)$ as the only nonzero element.

Let the matrices $P(b\mid a)$ with $a\in\Sigma$ and $b\in\Omega$ contain in each row at most one nonzero entry,
say $p(b,s_j\mid a,s_i)$. 
Then put $s_j=\gamma(b,a,s_i)$.
This defines a mapping $\gamma:\Omega\times\Sigma \times S\rightarrow S$ such that $A$ is observable. 
\end{proof}

\begin{proposition}
Let $A$ be a stochastic automaton and $\gamma:\Omega^*\times \Sigma^* \times S\rightarrow S$ be the mapping as above.
Then for all $s,s'\in S$, $x\in\Sigma^*$ and $y\in\Omega^*$ with $|x|=|y|$, 
\begin{eqnarray}
p(y,s'\mid x,s)\ne 0\quad\Rightarrow\quad s' = \gamma(y,x,s).
\end{eqnarray}
\end{proposition}
\begin{proof}
Let $y=by'$ and $x=ax'$ with $a\in \Sigma$, $b\in\Omega$, $x'\in\Sigma^*$ and $y'\in\Omega^*$ with $k=|x'|=|y'|$. 

For $k=0$ the assertion is clear. 
Let $k\geq 1$.
Suppose $p(by',s'\mid ax',s)\ne 0$.
Then there is a state $s''\in S$ such that 
$p(b,s''\mid a,s)\ne 0$ and  $p(y',s'\mid x',s'')\ne 0$.
By induction,
$s'' = \gamma(b,a,s)$ and $s' = \gamma(y',x',s'')$.
Hence, $s' = \gamma(y',x',\gamma(b,a,s)) ) =\gamma(by',ax',s)$.
\end{proof}

\begin{proposition}
Let $A$ and $B$ be observable stochastic automata with mappings $\gamma_A$ and $\gamma_B$, respectively.
For all $s\in S_A$, $s'\in S_B$, $a\in\Sigma$, and $b\in\Omega$,
if $p_A(b,s'\mid a,s)\ne 0$ and the states $s$ and $s'$ are equivalent, 
the states $\gamma_A(b,a,s)$ and $\gamma_B(b,a,s')$ are also equivalent.
\end{proposition}
\begin{proof}
Let $a\in\Sigma$ and $b\in \Omega$.
By Lemma~\ref{l-Ao},
the rows $s\cdot P_A(b\mid a)$ and $s'\cdot P_B(b\mid a)$ contain a single nonzero element
$p_A(b,\gamma_A(b,a,s)\mid a,s)$ and
$p_B(b,\gamma_B(b,a,s')\mid a,s')$, respectively.
These are exactly the entries of $s\cdot\eta_A(b\mid a)$ and $s'\cdot\eta_B(b\mid a)$, respectively.
Thus by the convention that single states represent unit vectors in the state space, we obtain
$$\frac{1}{s\cdot \eta_A(b\mid a)} s\cdot P_A(b\mid a) =\gamma_A(b,a,s),$$
where $s\cdot \eta_A(b\mid a) \ne0$, and
$$\frac{1}{s'\cdot \eta_B(b\mid a)} s'\cdot P_B(b\mid a) =\gamma_B(b,a,s'),$$
where $s'\cdot \eta_B(b\mid a) \ne0$.
Then for all $a\in\Sigma$, $x\in\Sigma^*$, $b\in \Omega$, and $y\in\Omega^*$, 
\begin{eqnarray*}
\gamma_A(b,a,s)\cdot\eta_A(y\mid x)
&=& \frac{1}{s\cdot \eta_A(b\mid a)} s\cdot P_A(b\mid a)\cdot\eta_A(y\mid x)\\
&=& \frac{1}{s\cdot \eta_A(b\mid a)} s\cdot \eta_A(by\mid ax), \; \mbox{by Prop.~\ref{p-omega}},\\
&=& \frac{1}{s'\cdot \eta_B(b\mid a)} s'\cdot \eta_B(by\mid ax), \; \mbox{by equivalence of $s,s'$},\\
&=& \frac{1}{s'\cdot \eta_B(b\mid a)} s'\cdot P_B(b\mid a)\; \eta_B(y\mid x), \quad \mbox{by Prop.~\ref{p-omega}},\\
&=& \gamma_B(b,a,s')\cdot\eta_B(y\mid x).
\end{eqnarray*}
Hence,  $\gamma_A(b,a,s)$ and $\gamma_B(b,a,s')$ are equivalent.
\end{proof}

\begin{example}
Take the stochastic automaton $A=(\{s_1,s_2,s_3\},\{a\},\{b,c\},p)$  in Fig.~\ref{f-Ao} with
$$
P(b\mid a) = \begin{pmatrix}
0 & \frac{1}{2} & 0\\
0 & \frac{1}{3} & 0\\
0 & \frac{5}{12} & 0\\
\end{pmatrix}
\quad\mbox{and}\quad
P(c\mid a) = \begin{pmatrix}
0 & \frac{1}{2} & 0\\
0 & \frac{2}{3} & 0\\
0 & \frac{7}{12} & 0\\
\end{pmatrix}.
$$
This automaton is observable with $\gamma(b,a,s) = \gamma(c,a,s) = s_2$ for all $s\in S$.
It is reduced, since the row vectors in the corresponding $H$-matrix are distinct (Prop.~\ref{p-simplex}),
$$H_A = \begin{pmatrix}
1 & \frac{1}{2} \\
1 & \frac{1}{3} \\
1 & \frac{5}{12}\\
\end{pmatrix}.
$$
However, the automaton is not minimal by Prop.~\ref{p-HA-pi},
since the state $s_3$ is equivalent to the state distribution $(\frac{1}{2},\frac{1}{2},0)$.
\EXX
\end{example}
\begin{figure}[hbt]
\begin{center}
\mbox{$
\xymatrix{
 *++[o][F-]{s_1} \ar@{->}[rrd]_{a,b:1/2;\; a,c:1/2} \\
&& *++[o][F-]{s_2} \ar@(ur,dr)[]^{a,b:1/3;\; a,c:2/3}\\
*++[o][F-]{s_3} \ar@{->}[urr]^{a,b:5/12;\;a,c:7/12}\\
}
$}
\end{center}
\caption{State diagram of observable stochastic automaton $A$.}\label{f-Ao}
\end{figure}
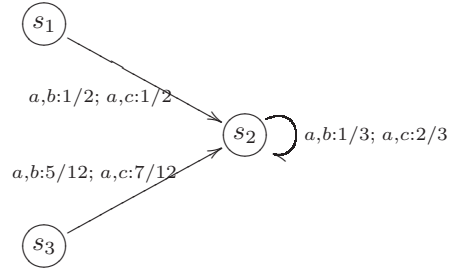

The last part of this chapter is devoted to deterministic automata.

A stochastic automaton $A$ is called {\em state-determined\/}\index{stochastic automaton!state-determined}
if there is a mapping $\delta:\Sigma\times S\rightarrow S$ such that for all
$a\in\Sigma$ and $s\in S$,
\begin{eqnarray}
\sum_{b\in \Omega} p(b,\delta(a,s)\mid a,s) = 1.
\end{eqnarray}
If a state-determined stochastic automaton $A$ is in state $s$ and $a$ is the input,
then $A$ transits into the state  $\delta(a,s)$ with certainty.

\begin{lemma}\label{l-As}
A stochastic automaton $A$ is state-determined if and only if
for each $a\in\Sigma$,
the matrix $P(a) = \sum_{b\in\Omega} P(b\mid a)$ contains in each row exactly one entry~$1$.
\end{lemma}
A state-determined stochastic automaton is a specialized observable stochastic automaton in which the
mapping $\gamma$ is independent of the output.
\begin{example}
Consider the stochastic automaton $A=(\{s_1,s_2\},\{a\},\{b,c\},p)$  in Fig.~\ref{f-Astate} with
$$
P(b\mid a) = \begin{pmatrix}
0 & \frac{1}{3} \\
0 & \frac{2}{3} \\
\end{pmatrix}
\quad\mbox{and}\quad
P(c\mid a) = \begin{pmatrix}
0 & \frac{2}{3} \\
0 & \frac{1}{3} \\
\end{pmatrix}.
$$
This automaton is state-determined and strongly reduced by Prop.~\ref{p-strongr}, 
since the $H$-matrix has full rank,
$$H_A = \begin{pmatrix}
1 & \frac{1}{3} \\
1 & \frac{2}{3} \\
\end{pmatrix}.
$$
\EXX
\end{example}
\begin{figure}[hbt]
\begin{center}
\mbox{$
\xymatrix{
*++[o][F-]{s_1} \ar@/^/[rr]^{a,b:1/3;\;a,c:2/3} 
&& *++[o][F-]{s_2} \ar@/^/[ll]^{a,b:2/3;\;a,c:1/3} 
}
$}
\end{center}
\caption{State diagram of state-determined stochastic automaton $A$.}\label{f-Astate}
\end{figure}
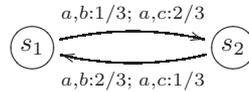

A stochastic automaton $A$ is called {\em output-determined\/}\index{stochastic automaton!output-determined}
if there is a mapping $\lambda:\Sigma\times S\rightarrow \Omega$ such that for all
$a\in\Sigma$ and $s\in S$,
\begin{eqnarray}
\sum_{s'\in S} p(\lambda(a,s),s'\mid a,s) = 1.
\end{eqnarray}
If an output-determined stochastic automaton $A$ is in state $s$ and $a$ is the input,
then $A$ outputs $\lambda(a,s)$ with certainty.

\begin{lemma}\label{l-Aoput}
A stochastic automaton $A$ is output-determined if and only if
for all $a\in\Sigma$ and $s\in S$, 
the matrix $P(\lambda(a,s)\mid a)$ contains a nonzero row labeled by~$s$ whose entries sum up to~$1$.
\end{lemma}

\begin{proposition}
For each output-determined stochastic automaton there exists an output-determined reduced stochastic automaton.
\end{proposition}
\begin{proof}
Construct for a given output-determined stochastic automaton $A$ a reduced stochastic automaton $B$ by
the powerset method described in the proof of Thm.~\ref{t-redA}.

Using the notation in the proof, two states $s$ and $t$ in $A$ are equivalent if and only if they lie an 
a state $Z_i$ of $B$.
In particular,
$s\cdot \eta_A(b\mid a) = t\cdot \eta_A(b\mid a)$.
If this value is nonzero, then the output is $b=\lambda_A(a,s)=\lambda_A(a,t)$, since $A$ is output-determined.
Therefore, the mapping $\lambda_B:\Sigma\times Z\rightarrow Z$ given by $\lambda_B(a,Z_i)=b$ if
$\lambda_A(a,s)=b$ for some $s\in Z_i$ is well-defined.
Moreover, $\sum_{Z_j\in Z}p_B(\lambda(a,Z_i),Z_j\mid a,Z_i) = 1$ for all $a\in\Sigma$ and $Z_i\in Z$.
Hence, $B$ is also output-determined.
\end{proof}

\begin{example}
Consider the stochastic automaton $A=(\{s_1,s_2\},\{a\},\{b,c\},p)$  
with $\lambda(a,s_1)=b$, $\lambda(a,s_2)=c$, and
$$
P(b\mid a) = \begin{pmatrix}
\frac{1}{3} & \frac{2}{3} \\
0 & 0\\
\end{pmatrix}
\quad\mbox{and}\quad
P(c\mid a) = \begin{pmatrix}
0 & 0 \\
\frac{1}{2} & \frac{1}{2} \\
\end{pmatrix}.
$$
This automaton is output-determined and strongly reduced by Prop.~\ref{p-strongr}, 
since the $H$-matrix has full rank,
$$H_A = \begin{pmatrix}
1 & 1 \\
1 & 0 \\
\end{pmatrix}.
$$
\EXX
\end{example}

A stochastic automaton $A$ is called {\em determined\/}\index{stochastic automaton!determined}
if there are mappings $\delta:\Sigma\times S\rightarrow S$ and $\lambda:\Sigma\times S\rightarrow \Omega$ 
such that for all $a\in\Sigma$ and $s\in S$,
\begin{eqnarray}
p(\lambda(x,s),\delta(a,s)\mid a,s) = 1.
\end{eqnarray}
If a determined stochastic automaton $A$ is in state $s$ and $a$ is the input, then
$A$ enters state $\delta(a,s)$ and outputs $\lambda(a,s)$ with certainty.
\begin{lemma}\label{l-Ad}
A stochastic automaton $A$ is determined if and only if for all $a\in\Sigma$ and $b\in \Omega$,
the matrix $P(b\mid a)$ has in each row at most one entry~$1$ and all other entries are~$0$.
\end{lemma}

\begin{proposition}
A stochastic automaton $A$ is determined if and only if $A$ is observable and output-determined.
\end{proposition}
\begin{proof}
The matrices $P(b\mid a)$ of a determined stochastic automaton have by Lemma~\ref{l-Ad} the same form as the
matrices of an observable and output-deter\-mined stochastic automaton by Lemmata~\ref{l-Ao} and~\ref{l-Aoput}.
\end{proof}

\begin{proposition}
A determined stochastic automaton $A$ is reduced if and only if\/ $A$ is minimal.
\end{proposition}
\begin{proof}
Let $A$ be determined and reduced and let $S=\{s_1,\ldots,s_n\}$.
Suppose the state $s\in S$ and the state distribution $\pi$ are equivalent.
Since $A$ is determined,
for each input $x\in\Sigma^*$ there is exactly one state $s'\in S$ and one output $y\in\Omega^*$ 
such that $p(y,s'\mid x,s)=1$.
Moreover, since $s$ and $\pi$ are equivalent,
$1=p(y,s'\mid x,s)=s\cdot \eta(y\mid x) = \pi\cdot\eta(y\mid x)$.
But $\pi\cdot\eta(y\mid x) = 1$ only holds if all components $\eta_i(y\mid x)$ are equal to~$1$ if $\pi_i\ne 0$.
Thus $\eta_i(y\mid x) = s_i\cdot \eta(y\mid x) = 1 = s\cdot\eta(y\mid x)$ and hence $s_i$ and $s$ are equivalent.
However, $A$ is reduced and so $s$ and $s'$ cannot be equivalent.
Hence, $A$ is minimal.

Conversely, each minimal stochastic automaton is also reduced.
\end{proof}

\begin{example}
Take the stochastic automaton $A=(\{s_1,s_2,s_3\},\{a\},\{b,c\},p)$  
with 
$$
P(b\mid a) = \begin{pmatrix}
1 & 0 & 0 \\
0 & 0 & 0 \\
\frac{1}{2} & 0 & 0 \\
\end{pmatrix}
\quad\mbox{and}\quad
P(c\mid a) = \begin{pmatrix}
0 & 0 & 0\\
0 & 1 & 0\\
0 & \frac{1}{2} & 0 \\
\end{pmatrix}.
$$
This automaton is observable by Lemma~\ref{l-Ao}.
The associated reduced stochastic automaton is $B =(\{t_1,t_2\},\{a\},\{b,c\},p)$  
with
$$
P(b\mid a) = \begin{pmatrix}
1 & 0  \\
0 & 0  \\
\end{pmatrix}
\quad\mbox{and}\quad
P(c\mid a) = \begin{pmatrix}
0 & 0 \\
0 & 1 \\
\end{pmatrix}.
$$
Moreover, $B$ is determined and therefore minimal.
\EXX
\end{example}

\chapter{Stochastic Mealy and Moore Automata}

Stochastic versions of the well-known Mealy and Moore automata are presented and
the generality of stochastic Moore automata is shown.

A stochastic automaton $A = (S,\Sigma,\Omega,p)$ is called 
{\em stochastic Mealy automaton}\index{stochastic Mealy automaton}
if there are conditional probabilities $p_1(\cdot\mid a,s)$ and $p_2(\cdot\mid a,s)$ over $\Omega$ and $S$, 
respectively, such that for all $a\in\Sigma$, $b\in\Omega$ and $s,s'\in S$,
\begin{eqnarray}
p(b,s'\mid a,s) = p_1(b\mid a,s)\cdot p_2(s'\mid a,s).
\end{eqnarray}
Stochastic independence of state transition and output emission makes stochastic Mealy automata a 
restricted class of stochastic automata. 

\begin{proposition}
Each state-determined stochastic automaton $A$ is a stochastic Mealy automaton.
Each output-determined stochastic automaton $A$ is a stochastic Mealy automaton.
\end{proposition}
\begin{proof}
Let $A$ be state-determined.
Define
$p_1(b\mid a,s) =  p(b,\delta(a,s)\mid a,s)$ and 
$p_2(s'\mid a,s) =  1$ if $\delta(a,s)=s'$ and $0$ otherwise.
Then $p(b,s'\mid a,s) = p_1(b\mid a,s) \cdot p_2(s'\mid a,s)$ as required.

Let $A$ be output-determined.
Define
$p_1(b\mid a,s) =  1$ if $b=\lambda(a,s)$ and $0$ otherwise, and
$p_2(s'\mid a,s) =  p(\lambda(a,s),s'\mid a,s)$.
Then $p(b,s'\mid a,s) = p_1(b\mid a,s) \cdot p_2(s'\mid a,s)$ as required.
\end{proof}

\begin{proposition}
A stochastic Mealy automaton $A$ is observable if and only if $A$ is state-determined.
\end{proposition}
\begin{proof}
Let $A$ be observable.
Then for all $a\in\Sigma$, $b\in\Omega$, and $s\in S$,
$$p(b,\gamma(b,a,s)\mid a,s) = p_1(b\mid a,s)\cdot p_2(\gamma(b,a,s)\mid a,s).$$
Thus $\gamma(b,a,s)$ is independent of $b$ and therefore we put $\delta(a,s)= \gamma(b,a,s)$.
Then
\begin{eqnarray*}
\sum_{b\in \Omega} p(b,\delta(a,s)\mid a,s) 
&=& \sum_{b\in \Omega} p_1(b\mid a,s) p_2(\delta(a,s)\mid a,s)  \\
&=& \sum_{b\in \Omega} p_1(b\mid a,s)  = 1.
\end{eqnarray*}
Hence, $A$ is state-determined.
The converse is similar.
\end{proof}

\begin{proposition}
For each stochastic Mealy automaton $A$ there is a reduced stochastic Mealy automaton.
\end{proposition}
\begin{proof}
Construct for a given stochastic Mealy automaton $A$ a reduced stochastic automaton $B$ by
the powerset method described in the proof of Thm.~\ref{t-redA}.

Using the notation in the proof, 
we obtain for all $a\in\Sigma$, $b\in\Omega$, and $1\leq i,j\leq r$,
\begin{eqnarray*}
p_B(b,Z_j\mid a, Z_i) 
&=& \sum_{z\in Z_j}p_A(b,z,\mid a,z_i)\\
&=& \sum_{z\in Z_j}p_{1,A}(b\mid a,z_i) p_{2,A}(z\mid a,z_i).
\end{eqnarray*}
Define 
$$p_{1,B}(b\mid y,Z_i) = p_{1,A}(b\mid a,z_i)$$
and 
$$p_{2,B}(Z_j\mid y,Z_i) = \sum_{z\in Z_j} p_{2,A}(z\mid a,z_i).$$
These values are independent of the chosen representative state $z_i\in Z_i$.
Thus
$$p_B(b,Z_j\mid a, Z_i) = p_{1,B}(b\mid y,Z_i)\cdot p_{2,B}(Z_j\mid y,Z_i)$$
and hence $B$ is a stochastic Mealy automaton.
\end{proof}

\begin{example}
Take the stochastic Mealy automaton
$A=(\{s_1,s_2\},\{a\},\{b,c\},p)$ with probabilities
$$
\begin{array}{c|cccc}
& p_1(b\mid a,\cdot ) & p_1(c\mid a,\cdot )   & p_2(s_1\mid a,\cdot ) & p_2(s_2\mid a,\cdot )   \\\hline
s_1 & \frac{1}{2} & \frac{1}{2} & \frac{1}{3} & \frac{2}{3} \\
s_2 & \frac{1}{4} & \frac{3}{4} & \frac{1}{5} & \frac{4}{5} \\
\end{array}
$$
Then we have
$$
P(b\mid a) = \begin{pmatrix}
\frac{1}{6} & \frac{1}{3}  \\
\frac{1}{20} & \frac{1}{5}  \\
\end{pmatrix}
\quad\mbox{and}\quad
P(c\mid a) = \begin{pmatrix}
\frac{1}{6} & \frac{1}{3} \\
\frac{3}{20} & \frac{3}{5} \\
\end{pmatrix}.
$$
Therefore,
$$
P(a) = \begin{pmatrix}
\frac{1}{3} & \frac{2}{3}  \\
\frac{1}{5} & \frac{4}{5}  \\
\end{pmatrix}.
$$
\EXX
\end{example}

A stochastic automaton $A = (S,\Sigma,\Omega,p)$ is called 
{\em stochastic Moore automaton\/}\index{stochastic Moore automaton}
if there are conditional probabilities $\mu(\cdot\mid s)$ and $p'(\cdot\mid a,s)$ over $\Omega$ and $S$, 
respectively, such that for all $a\in\Sigma$, $b\in\Omega$ and $s,s'\in S$,
\begin{eqnarray}
p(b,s'\mid a,s) = \mu(b\mid s')\cdot p'(s'\mid a,s).
\end{eqnarray}
In a Moore automaton the output emission depends on the state transition.
\begin{theorem}
For each stochastic automaton $A = (S_A,\Sigma,\Omega,p_A)$ 
there is a stochastic Moore automaton $B = (S_B,\Sigma,\Omega,p_B)$ such that $A$ and $B$ are S-equivalent.
\end{theorem}
\begin{proof}
By Prop.~\ref{p-S-hom}, it is sufficient to construct an S-epimorphism $\phi:B\rightarrow A$.
For this, put $S_B=\Omega\times S_A$ and
$$
p_B(b_2,(b_1,s_1)\mid a,(b_0,s_0)) 
=\left\{ \begin{array}{ll}
p_A(b_1,s_1\mid a,s_0) & \mbox{if } b_1=b_2,\\
0 & \mbox{otherwise}
\end{array} \right.
$$
for all $a\in\Sigma$, $b_2\in\Omega$, and $(b_1,s_1),(b_0,s_0)\in S_B$.

Moreover, for all $a\in\Sigma$ and $(b_1,s_1),(b_0,s_0)\in S_B$ put
$$p'((b_1,s_1)\mid a,(b_0,s_0)) = p_A(b_1,s_1\mid a,s_0)$$ 
and
$$
\mu(b_2\mid (b_1,s_1)) =\left\{ \begin{array}{ll}
1 & \mbox{if } b_1=b_2,\\
0 & \mbox{otherwise.}
\end{array} \right.
$$
Then we have for all $a\in\Sigma$, $b_2\in\Omega$, and $(b_1,s_1),(b_0,s_0)\in S_B$,
$$ p_B(b_2,(b_1,s_1)\mid a,(b_0,s_0)) 
= \mu(b_2\mid (b_1,s_1))\cdot p'((b_1,s_1)\mid a,(b_0,s_0)).$$
It is clear that $B$ is a stochastic Moore automaton.
In particular, $\mu$ is a mapping, i.e., the successor state determines uniquely the output symbol.

Finally, define the mapping $\phi:S_B\rightarrow S_A$ by $\phi((b,s))  = s$  for all $(b,s)\in S_B$.
Then we have
\begin{eqnarray*}
p_B(b_2,\phi^{-1}\phi(b_1,s_1)\mid a,(b_0,s_0)) 
&=& \sum_{b\in \Omega} p_B(b_2,(b,s_1)\mid a,(b_0,s_0))\\
&=& p_B(b_2, (b_2,s_1)\mid a,(b_0,s_0))\\
&=& p_A(b_2,s_1\mid a,s_0)\\
&=& p_A(b_2,\phi(b_1,s_1) \mid a,\phi(b_0,s_0)).
\end{eqnarray*}
Hence, $\phi$ is an S-epimorphism.
\end{proof}

The above result shows that stochastic Moore automata already provide the most general type of stochastic automata.  However, this is not true for stochastic Mealy automata, since the transition probabilities for successor state and output are independent of each other given input and state.

\chapter{Stochastic Acceptors}
Stochastic acceptors are a generalization of the nondeterministic finite automata.
The languages recognized by stochastic acceptors are called stochastic languages.
We will see that the class of stochastic languages is uncountable and includes the regular languages.
For this, we assume familiarity with the basic concepts of regular languages as well as deterministic 
and nondeterministic finite acceptors~\cite{salomaa}.

A {\em stochastic acceptor\/}\index{stochastic automaton!acceptor}\index{stochastic acceptor}
is a quintuple $A = (S,\Sigma,P, \pi,f)$, where
\begin{itemize}
\item $S$ is a nonempty finite set of {\em states},
\item $\Sigma$ is an alphabet of {\em input symbols},
\item $P$ is a collection $\{P(a)\mid a\in\Sigma\}$ of stochastic $n\times n$ matrices, 
where $n$ is the number of states,
\item $\pi$ is an {\em initial distribution}\index{initial distribution} of the states written as row vector,
\item $f$ is a binary column vector of length $n$ called {\em final state vector}\index{final state vector}.
\end{itemize}
Note that if the state set is $S=\{s_1,\ldots,s_n\}$ and the final state vector is $f=(f_1,\ldots,f_n)^T$, then 
$F = \{s_i\mid f_i=1\}$ is the {\em final state set}\index{final state set}.

The matrices $P(a)=(p_{ij}(a))$ with $a\in\Sigma$ are transition probability matrices, 
where the $(i,j)$th entry $p_{ij}(a) = p(s_j\mid a,s_i)$ is the conditional probability of transition from state $s_i$ to state $s_j$ when the symbol $a$ is read, 
$1\leq i,j\leq n$.
A state change from one state to another must take place with probability~1.
Thus for each symbol $a\in \Sigma$ and each state $s_i\in S$, we have
\begin{eqnarray}
\sum_{j=1}^n p(s_j\mid a,s_i) = 1.
\end{eqnarray}

Given the conditional probability distribution $p(\cdot\mid a,s)$ on $S$,  
we define a conditional probability distribution $\hat p$ recursively as follows:
\begin{itemize}
\item For all states $s,s'\in S$,
\begin{eqnarray}\label{e-SAA-phat1}
\hat p (s'\mid \epsilon,s) = \left\{ \begin{array}{ll} 1 & \mbox{if } s=s',\\ 0 & \mbox{if } s\ne s',  \end{array} \right.
\end{eqnarray}
where $\epsilon$ denotes the empty word in $\Sigma^*$.
\item For all $s'\in S$, $a\in \Sigma$, and $x\in\Sigma^*$,
\begin{eqnarray}\label{e-SAA-phat3}
\hat p (s'\mid ax,s) = \sum_{t\in S} p(t\mid a,s)\cdot \hat p(s'\mid x,t).
\end{eqnarray}
\end{itemize}
Then $\hat p(\cdot\mid x,s)$ is a conditional probability distribution on $S$ and so we have
\begin{eqnarray}\label{e-SAA-phat4}
\sum_{s'\in S} \hat p(s'\mid x,s) =1,\quad x\in\Sigma^*, \,s\in S.
\end{eqnarray}
Note that the condititional probability distributions $p(\cdot\mid a,s)$ and $\hat p(\cdot\mid a,s)$ coincide 
on $\Sigma\times S$ and therefore we write $p$ instead of $\hat p$.

A stochastic acceptor works serially and synchronously.
It reads an input word symbol by symbol and after reading an input symbol it transits from one state into another.
In particular, if the automaton starts in state $s$ and reads the word $x$, then
with probability $p(s'\mid x,s)$ it will end in state $s'$ taking all intermediate states into account.

\begin{proposition}\label{p-SAA-prob0}
For all $x,x'\in\Sigma^*$, and $s,s'\in S$,
$$p(s'\mid xx',s) = \sum_{t\in S} p(t\mid x,s) \cdot p(s'\mid x',t).$$ 
\end{proposition}
This result is a special case of Prop.~\ref{p-SA-prob0}.

The behavior of a stochastic acceptor can be described by transition probability matrices.
To see this, let $A$ be a stochastic acceptor with state set $S=\{s_1,\ldots,s_n\}$.
For the empty word, define by~(\ref{e-SAA-phat1}), 
\begin{eqnarray}
P(\epsilon) = I_n,
\end{eqnarray}
where $I_n$ is the $n\times n$ unit matrix.
Furthermore, if $a\in \Sigma$ and $x\in\Sigma^*$, then by~(\ref{e-SAA-phat3}) we have
\begin{eqnarray}
P(ax) = P(a)\cdot P(x).
\end{eqnarray}
By Prop.~\ref{p-SAA-prob0} and the associativity of matrix multiplication, we obtain the following.
\begin{proposition}\label{p-SAA-prob1}
For all $x,x'\in\Sigma^*$, 
$$P(xx') = P(x) \cdot P(x').$$ 
\end{proposition}
It follows by induction that if $x=x_1\ldots x_k\in\Sigma^*$, then
\begin{eqnarray}\label{e-SAA-prob2}
P(x) = P(x_1)\cdots P(x_k).
\end{eqnarray}
The $(i,j)$-th element of the product matrix $P(x)$
is the probability of transition from state $s_i$ to state $s_j$ when the word $x$ is read.
Moreover, if $\pi$ is the initial state distribution of the stochastic acceptor, then after reading
the word $x\in \Sigma^*$ the final state distribution of the stochastic acceptor becomes $\pi P(x)$.
\begin{example}
Take the stochastic acceptor $A = (\{s_1,s_2\},\{a,b\}, P, \pi,f)$ with
$$ P(a) = \left( \begin{array}{cc} \frac{1}{2} &\frac{1}{2}\\ 0 & 1 \end{array} \right),
\quad P(b) = \left( \begin{array}{cc} \frac{1}{4} &\frac{3}{4}\\ \frac{1}{8} & \frac{7}{8} \end{array} \right),
\quad \pi = (1,0), 
\quad\mbox{and}\quad f= \left( \begin{array}{l} 0\\ 1 \end{array} \right).  $$
Then we have for all input words of length~3,
\begin{eqnarray*}
\begin{array}{lll}
\pi P(aaa) f = \frac{7}{8},    && \pi P(baa) f =  \frac{15}{16},\\
\pi P(aab) f =  \frac{27}{32}, && \pi P(bab) f =  \frac{55}{64},\\
\pi P(aba) f =  \frac{29}{32}, && \pi P(bba) f =  \frac{59}{64},\\
\pi P(abb) f =  \frac{109}{128}, && \pi P(bbb) f =  \frac{219}{256}.
\end{array}
\end{eqnarray*}
\EXX
\end{example}

A stochastic acceptor $A=(S,\Sigma,P,\pi,f)$ is {\em deterministic\/}\index{stochastic acceptor!deterministic}
if each matrix $P(a)$, $a\in\Sigma$, has in each row exactly one entry~1.
Note that a deterministic stochastic acceptor is a nondeterministic acceptor in the usual sense if the initial state distribution is a single state.
\begin{example}
Take the deterministic acceptor $A = (\{s_1,s_2\},\{a,b\}, P, \pi,f)$ with
$$ P(a) = \left( \begin{array}{cc} 1&0\\ 0 & 1 \end{array} \right),
\quad P(b) = \left( \begin{array}{cc} 0&1\\ 1 & 0 \end{array} \right),
\quad \pi = (1,0), 
\quad\mbox{and}\quad f= \left( \begin{array}{l} 0\\ 1 \end{array} \right).  $$
Then we have for all input words of length~3,
\begin{eqnarray*}
\begin{array}{lll}
\pi P(aaa) f = 0, && \pi P(baa) f = 1,\\
\pi P(aab) f = 1, && \pi P(bab) f = 0,\\
\pi P(aba) f = 1, && \pi P(bba) f = 0,\\
\pi P(abb) f = 0, && \pi P(bbb) f = 1.\\
\end{array}
\end{eqnarray*}
\EXX
\end{example}

Let $A=(S,\Sigma,P,\pi,f)$ be a stochastic acceptor and let $\lambda$ be a real number with $0\leq \lambda\leq 1$.
The set 
\begin{eqnarray}
L_{A,\lambda} = \{x\in \Sigma^*\mid \pi P(x)f > \lambda\}
\end{eqnarray}
is the {\em stochastic language\/}\index{language!stochastic} of the acceptor 
and the number $\lambda$ is the {\em cut point\/}\index{cut point} 
of the language. 
Note that by definition $L_{A,1} = \emptyset$.

\begin{example}\label{e-SA-0}
Consider the stochastic acceptor $A=(\{s_1,s_2\},\{a\},P,\pi,f)$ with
$$ P(a) = \left( \begin{array}{cc} \frac{1}{2} &\frac{1}{2}\\ 0 & 1 \end{array} \right),
\quad \pi = (1,0), \quad\mbox{and}\quad 
f= \left( \begin{array}{l} 0\\ 1 \end{array} \right).  $$
Then for each integer $k\geq 0$,
$$\pi P(a^k) f 
= \pi  P(a)^k f = \pi \left( \begin{array}{cc} \frac{1}{2^k} &\frac{2^k-1}{2^k}\\ 0 & 1 \end{array} \right) f
= 1-\frac{1}{2^k}.$$
Given the cut point $\lambda$ with $1-\frac{1}{2^{l+1}}> \lambda \geq 1-\frac{1}{2^l}$ for some $l\geq 0$, the accepted language is
$$L_{A,\lambda} = \{a^k \mid  1-\frac{1}{2^k} > \lambda \} = \{a^k \mid  k> l\}.$$
This language is regular for each cut point $\lambda$.
\EXX
\end{example}

A subset $L$ of $\Sigma^*$ is a {\em $\lambda$-stochastic language}\index{language!lambda-stochastic} 
with $0\leq \lambda\leq 1$
if there is a stochastic acceptor $A$ such that $L=L_{A,\lambda}$.
A subset $L$ of $\Sigma^*$ is a {\em stochastic language}\index{language!stochastic} 
if $L$ is $\lambda$-stochastic for some $0\leq \lambda\leq 1$.
A language $L$ is {\em $\lambda$-regular}\index{language!lambda-regular} with $0\leq \lambda\leq 1$
if there is a deterministic acceptor $A$ such that $L=L_{A,\lambda}$.
A language $L$ is {\em regular}\index{language!regular}
if $L$ is $\lambda$-regular for some $0\leq \lambda\leq 1$.

%
\begin{proposition}\label{p-sa-0}
Each $0$-stochastic language is regular.
\end{proposition}
\begin{proof}
Let $L$ be a $0$-stochastic language over $\Sigma$.
Then there is a stochastic acceptor 
$$A=(S_A,\Sigma,P_A	,\pi_A,f_A)$$ such that $L=L_{A,0}$.
Consider the deterministic acceptor 
$$B=(S_B,\Sigma,P_B,\pi_B,f_B),$$ where
\begin{itemize}
\item $S_B= 2^{S_A}$ is the power set of $S_A$.
\item The transition probability matrices $P_B(a)$, $a\in\Sigma$, are defined by
\begin{eqnarray*}
\lefteqn{p_B(S_j\mid a,S_i) }\\
&=&\left\{ \begin{array}{ll}
1 & \mbox{if } S_j = \{s'\in S_A \mid p_A(s'\mid a,s)\ne 0 \mbox{ for some } s\in S_i\},\\
0 & \mbox{otherwise,}
\end{array} \right.  
\end{eqnarray*}
for all $a\in \Sigma$ and $S_i,S_j\in S_B$.
\item
$\pi_B = (\pi_{B,1},\ldots,\pi_{B,2^n})$ with 
$$\pi_{B,i}= 
\left\{ \begin{array}{ll}
1 & \mbox{if } S_i = R,\\
0 & \mbox{otherwise,}
\end{array} \right.$$
where $S_B=\{S_1,\ldots,S_{2^n}\}$, $\pi_A = (\pi_{A,1},\ldots,\pi_{A,n})$, and 
$$R=\{s_i\in S_A\mid\pi_{A,i}\ne 0\}.$$
\item $f_B$ is uniquely determined by the final state set 
$$F_B=\{ R\in S_B\mid R \cap F_A\ne\emptyset\}.$$ 
\end{itemize}
Clearly, the acceptor $B$ is deterministic.
Moreover, for each $x\in \Sigma^*$, 
the inequality $\pi_A P_A(x) f_A>0$ 
is equivalent to $\pi_B P_B(x) f_B>0$. 
Hence, $L_{A,0} = L_{B,0}$. 
\end{proof}

\begin{example}
In view of the stochastic acceptor $A$ in Ex.~\ref{e-SA-0}, we have $L_{A,0}=\{a^k\mid k\geq 1\}$.
The corresponding deterministic acceptor $B$ has state set $S_B=\{\emptyset,\{s_1\},\{s_2\},\{s_1,s_2\}\}$, 
input alphabet $\Sigma_B=\{0,1\}$, state transition matrix
$$P_B(x) = \bordermatrix{ 
~              & \emptyset  & \{s_1\} & \{s_2\} & \{s_1,s_2\} \cr
\emptyset      &    1       &    0    &    0    &    0        \cr
\{s_1\}        &    0       &    0    &    0    &    1        \cr
\{s_2\}        &    0       &    0    &    1    &    0        \cr
\{s_1,s_2\}    &    0       &    0    &    0    &    1        \cr
}$$
as well as initial state distribution $\pi_B=(0,1,0,0)$, and final state set $F_B=\{\{s_2\},\{s_1,s_2\}\}$.
\EXX
\end{example}

\begin{proposition}\label{p-sa-pre-0}
For each deterministic acceptor $A$  and each cut point $\lambda$, there is a deterministic acceptor $B$ with 
initial state distribution given by a single state such that
$L_{A,\lambda} = L_{B,0}$. 
\end{proposition}
\begin{proof}
Let $A=(S_A,\Sigma,P_A,\pi_A,f_A)$ be a deterministic acceptor with 
$S_A=\{s_1,\ldots,s_n\}$ and $\pi_A=(\pi_{A,1},\ldots,\pi_{A,n})$.
Write $\pi = \sum_{i=1}^n \pi_{A,i}s_i$.

For each $1\leq i\leq n$,
define the acceptor $A_i$ obtained from $A$ by replacing the initial state distribution $\pi_A$ with the 
state (standard basis vector) $s_i$.
This acceptor is deterministic in the ordinary sense.

Let $x\in L_{A,\lambda}$. Then $\pi_A P_A(x) f_A> \lambda$ or equivalently 
$$\sum_{i=1}^n \pi_{A,i}s_iP_A(x)f_A>\lambda.$$
That is, there are states
$s_{i_1},\ldots,s_{i_k}$ such that $\sum_{j=1}^k \pi_{A,i_j}>\lambda$ and $x\in L_{A_{i_j},0}$ for all $1\leq j\leq k$. 
Equivalently, we have 
$$x\in \bigcup ( L_{A_{i_1},0}\cap\ldots\cap L_{A_{i_k},0})= L',$$
where the union extends over all nonempty subsets $\{s_{i_1},\ldots,s_{i_k}\}$ of $S_A$ such that $\sum_{j=1}^k \pi_{A,i_j}>\lambda$.
Thus the set $L' =L_{A,\lambda}$ is a finite union of intersections of regular languages. 
But the class of regular languages is closed under finite unions and intersections.
Thus the language $L'$ is regular in the ordinary sense and hence there exists a deterministic acceptor $B$ with
the desired property.
\end{proof}

The $p$-adic languages provide an example of stochastic languages which are not regular and also demonstrate that the class of stochastic languages is uncountable.
To see this, let $p\geq 2$ be an integer.
The stochastic acceptor $A= (\{s_1,s_2\}, \{0,\ldots,p-1\},  P , \pi, f)$ with
$$P(a) = 
\left(\begin{array}{cc} 1-\frac{a}{p} &\frac{a}{p} \\ 1-\frac{a+1}{p} &\frac{a+1}{p} \end{array}\right), \quad 0\leq a\leq p-1, 
\quad 
\pi =(1,0),
\quad \mbox{and}\quad 
f = \left( \begin{array}{c} 0\\1 \end{array}\right)$$
is called {\em $p$-adic acceptor}\index{acceptor!p-adic}.

Each word $x=x_1\ldots x_k\in\{0,\ldots,p-1\}^*$ can be assigned the real number 
whose $p$-adic representation is $0.x_k\ldots x_1$.

\begin{proposition}
Let $A$ be an $p$-adic acceptor and let $\lambda$ be a cut point.
Then 
$$L_{A,\lambda} = \{x_1\ldots x_k\in\{0,\ldots,p-1\}^*\mid 0.x_k\ldots x_1 > \lambda\}.$$
\end{proposition}
\begin{proof}
A word $x=x_1\ldots x_k\in\{0,\ldots,p-1\}^*$ lies in $L_{A,\lambda}$ 
if and only if
$\pi P(x) f > \lambda$.
Note that $\pi P(x) f$ is the upper right entry of the matrix $P(x)$.

Claim that the matrix $P(x)$ has $0.x_k\ldots x_1$ as upper right entry.
Indeed, this is true for $k=1$.
The matrix $P(x_1\ldots x_ka)$ with $x_1,\ldots,x_k,a \in\{0,\ldots,p-1\}$ can be written as the product matrix $P(x_1\ldots x_k)P(a)$.
By induction, the upper right entry of this matrix is 
\begin{eqnarray*}
(1-0.x_k\ldots x_1) \cdot \frac{a}{p} + 0.x_k\ldots x_1 \cdot \frac{a+1}{p} 
&=& \frac{a}{p} + \frac{0.x_k\ldots x_1}{p} \\
&=& 0.a + 0.0x_k\ldots x_1 \\
&=& 0.ax_k\ldots x_1.
\end{eqnarray*}
This proves the claim.
Hence, $x_1\ldots x_k\in L_{A,\lambda}$ if and only if $0.x_k\ldots x_1>\lambda$.
\end{proof}

The following result makes use of the {\em Nerode equivalence relation\/}\index{Nerode relation} 
$\equiv_L$ of a language $L$ over $\Sigma$.
Define for all words $x,y\in\Sigma^*$, 
\begin{eqnarray}
x\equiv_L y \quad:\Longleftrightarrow\quad
\forall z\in\Sigma^*: xz\in L \Longleftrightarrow yz\in L.
\end{eqnarray}
This is an equivalence relation on the set $\Sigma^*$.
The theorem of Nerode-Myhill says that a language $L\subseteq\Sigma^*$ is regular if and only if the number of equivalence 
classes of the relation $\equiv_L$ is finite.

For instance, the language $L=\{a^nb^n\mid n\in\NN\}$ is not regular, since 
the words $a^ib$ and $a^jb$ with $i\ne j$ are not equivalent and so the
equivalence classes of the words $a^ib$, $i\geq 1$, are all distinct.

\begin{proposition}\label{p-SA-rat}
Let $A$ be an $p$-adic acceptor and $\lambda$ be a cut point.
Then $L_{A,\lambda}$ is regular if and only if\/  $\lambda$ is rational.
\end{proposition}
\begin{proof}
In view of the Nerode equivalence relation of the language $L=L_{A,\lambda}$,
two words $x,y\in\Sigma^*$ belong to different equivalence classes if and only if
there is a word $z\in\Sigma^*$ such that $xz\in L$ and $yz\not\in L$ or $xz\not\in L$ and $yz\in L$.
That is, $0.\mi(yz)\leq \lambda < 0.\mi(xz)$ or $0.\mi(xz)\leq \lambda < 0.\mi(yz)$,
where $\mi(x_1\ldots x_k) = x_k\ldots x_1$ denotes the {\em mirror image}\index{mirror image} of the word $x=x_1\ldots x_k$.

Suppose without restriction that $0.\mi(yz)\leq \lambda < 0.\mi(xz)$. 
Since $\mu(yz)=\mi(z)\mi(y)$ and $\mu(xz)=\mi(z)\mi(x)$, we have
$0.\mi(z)\mi(y)\leq \lambda < 0.\mi(z)\mi(x)$. 
Thus the cut point must have the form $\lambda = 0.\mi(z)a_1a_2\ldots$.
If we put $\lambda_z = 0.a_1a_2\ldots$, then $0.\mi(y)\leq \lambda_z<0.\mi(x)$.
With this kind of numbers $\lambda_z$ it is possible to separate the equivalence classes.
That is, the relation $\equiv_L$ has the classes
$$\{y\in\Sigma^*\mid \lambda_{u}<0.\mi(y)\leq \lambda_v\mbox{ and there is no $w\in\Sigma^*$ with }\lambda_u<\lambda_w<\lambda_v\}.$$
We have 
$$\lambda = 0.\mi(z) +\frac{\lambda_z}{p^l},$$ 
where $l$ is the length of the word $z$.
Therefore, there are only finitely many equivalence classes of $\equiv_L$ if and only if 
there are only finitely many $\lambda_z$ with this property.
In this case, the cut point $\lambda$ is eventually periodic 
and hence a rational number.
\end{proof}
By definition, each regular language is stochastic.
However, if the cut point $\lambda$ is not rational, one obtains a stochastic language which is not regular. 
For this, note that the set of algorithms written for a Turing machine is countable, while the set of real numbers is uncountable.
Moreover, for distinct cut points $\lambda_1$ and $\lambda_2$, the corresponding $p$-adic stochastic acceptors yield different stochastic languages.  
Thus the number of stochastic languages is uncountable.
Hence, there are stochastic languages that cannot be accepted by a Turing machine.

\begin{example}\label{e-SA-rat}
Take the stochastic acceptor $A=(\{s_1,\ldots,s_5\},\{a,b\},P,\pi,f)$, where
$$
P(a)= \left( \begin{array}{ccccc}
\frac{1}{2} &  \frac{1}{2} & 0 & 0 & 0 \\
0           &  1           & 0 & 0 & 0 \\
0           &  0           & 0 & 0 & 1 \\
0           &  0           & 0 & 0 & 1 \\
0           &  0           & 0 & 0 & 1 
\end{array} \right),
\quad
P(b)= \left( \begin{array}{ccccc}
0 & 0 & \frac{1}{2} & \frac{1}{2} & 0 \\
0 & 0 & \frac{1}{2} &    0        & \frac{1}{2}  \\
0 & 0 & \frac{1}{2} & \frac{1}{2} & 0 \\
0 & 0 &      0      & 1           & 0 \\
0 & 0 &      0      & 0           & 1 
\end{array} \right),$$
and
$$\pi = \left( \begin{array}{ccccc}
1 & 0 & 0 & 0 & 0 
\end{array} \right), 
\quad
f = \left( \begin{array}{c}
0\\ 0 \\ 0\\ 1\\ 0
\end{array} \right).$$
The state diagram of the acceptor is given in Fig.~\ref{f-SA-As}. 
For integers $i\geq 0$ and $j\geq 1$, we obtain
$$
P(a^i)= \left( \begin{array}{ccccc}
\frac{1}{2^i} &  \frac{2^i-1}{2^i} & 0 & 0 & 0 \\
0           &  1           & 0 & 0 & 0 \\
0           &  0           & 0 & 0 & 1 \\
0           &  0           & 0 & 0 & 1 \\
0           &  0           & 0 & 0 & 1 
\end{array} \right),
\quad
P(b^j)= \left( \begin{array}{ccccc}
0 & 0 & \frac{1}{2^j} & \frac{2^j-1}{2^j} & 0 \\
0 & 0 & \frac{1}{2^j} &    \frac{2^{j-1}-1}{2^j} & \frac{1}{2}  \\
0 & 0 & \frac{1}{2^j} & \frac{2^j-1}{2^j} & 0 \\
0 & 0 &      0      & 1           & 0 \\
0 & 0 &      0      & 0           & 1 
\end{array} \right).
$$
Thus
$$\pi P(a^ib^j) = \left( \begin{array}{ccccc}
0, & 0, & \frac{1}{2^j}, & \frac{(2^j-1)+(2^i-1)(2^{j-1}-1)}{2^{i+j}}, & \frac{2^i-1}{2^{i+1}} 
\end{array} \right) $$
and hence
$$\pi P(a^ib^j) f = \frac{1}{2} +\left(\frac{1}{2} \right)^{i+1} -\left(\frac{1}{2} \right)^j.$$
The state diagram shows that for each word $x$, which is not of the form $a^ib^j$ with $i\geq0$ and $j\geq 1$, 
we obtain $\pi P(x) f= 0$.
Therefore, we obtain the language 
$$L_{A,\frac{1}{2}}= \{a^ib^j\mid i\geq 0, j\geq i+2\},$$ 
which is not regular by the theorem of Nerode-Myhill.
\EXX
\end{example}
\begin{figure}[hbt]
\begin{center}
\mbox{$
\xymatrix{
                                                              &*++[o][F-]{s_2} \ar@(ul,ur)[]^{a:1}\ar@{->}[dr]^{b:1/2} \ar@{->}[d]^{b:1/2} & \\
*++[o][F-]{s_1}\ar@(ul,dl)[]_{a:1/2}\ar@{->}[ur]^{a:1/2}\ar@{->}[r]^{b:1/2}\ar@{->}[dr]_{b:1/2} & *++[o][F-]{s_3}\ar@(r,dr)[]^{b:1/2}\ar@{->}[r]^{a:1}\ar@{->}[d]^{b:1/2}                      & *++[o][F-]{s_5} \ar@(ur,dr)[]^{a,b:1} \\
                                                              &*++[o][F-]{s_4} \ar@(dl,dr)[]_{b:1} \ar@{->}[ur]_{a:1} & \\
}
$}
\end{center}
\caption{State diagram of stochastic acceptor $A$.}\label{f-SA-As}
\end{figure}
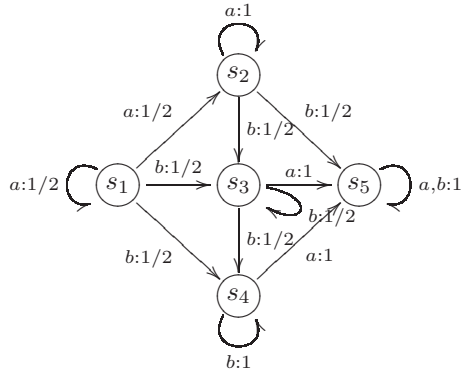

This example demonstrates that there are stochastic acceptors $A$, whose transition matrices 
and initial distributions are given by rational numbers, 
and rational cut points $\lambda$ such the language $L_{A,\lambda}$ is not regular.

Two normalization properties will come next.
First, stochastic languages are independent of the specific cut point.
\begin{proposition}
Let $A$ be a stochastic acceptor and $\lambda$ be a cut point.
Then for each number $\mu$ with $0<\mu<1$, there is a stochastic acceptor $B$ such that 
$L_{A,\lambda} = L_{B,\mu}$.
\end{proposition}
\begin{proof}
Let $A=(S,\Sigma,P,\pi,f)$ be a stochastic acceptor with state set 
$S_A = \{s_1,\ldots,s_n\}$ and let $\lambda$ be a cut point.

Let $\lambda=0$.
Then by Prop.~\ref{p-sa-0}, $L_{A,0}$ is regular 
and thus there is a deterministic acceptor $B$ which accepts this language for each cut point $\mu$ with $0<\mu<1$.

Let $\lambda=1$.
Then we have $L_{A,1}=\emptyset$ and it is clear that there is a stochastic acceptor $B$ with the required property.

Let $0<\lambda<1$.
Define the stochastic acceptor $B$ with alphabet $\Sigma$, 
state set $S_B = S_A\cup\{s_{n+1}\}$, and transition matrices
$$P_B(a) = \left(\begin{array}{ccc|c}
  &        &   & 0\\
  &P_A(a)  &   &  \vdots\\
  &        &   & 0\\\hline
0 & \ldots & 0 & 1
\end{array}\right),\quad a\in\Sigma.  $$
Then for each $x\in\Sigma^*$,
$$P_B(x) = \left(\begin{array}{ccc|c}
  &        &   & 0\\
  &P_A(x)  &   &  \vdots\\
  &        &   & 0\\\hline
0 & \ldots & 0 & 1
\end{array}\right).$$

Consider a cut point $\mu$ with $0<\mu<1$.
We will define a number $\nu$ with $0\leq \nu<1$ corresponding to $\lambda$ and $\mu$.
Put 
$$\pi_B = ((1-\nu)\pi_{A,1}, \ldots, (1-\nu)\pi_{A,n},\nu).$$
First, let $0<\mu\leq \lambda$.
Then put $\nu = 1-\frac{\mu}{\lambda}$ and 
$f_B = \left(\begin{array}{c} f_A\\ 0\end{array}\right)$.
This defines a stochastic acceptor $B$ with the property that for each word $x\in\Sigma^*$,
$$\pi_B P_B(x) f_B = (1-\nu)\pi_A P_A(x) f_A =\frac{\mu}{\lambda} \pi_A P_A f_A.$$
Thus $\pi_A P_A(x) f_A>\lambda$ if and only if $\pi_B P_B(x) f_B > \mu$.
Hence, $L_{A,\lambda} = L_{B,\mu}$.

Second, let $\lambda<\mu<1$.
Then put $\nu = \frac{\mu-\lambda}{1-\lambda}$ and
$f_B = \left(\begin{array}{c} f_A\\ 1\end{array}\right)$.
This defines a stochastic acceptor $B$ with the property that for each word $x\in\Sigma^*$,
$$\pi_B P_B(x) f_B = (1-\nu)\pi_AP_A(x)f_A + \nu.$$
We have $(1-\nu)\lambda + \nu = \mu$. 
Thus $\pi_A P_A(x) f_A>\lambda$ if and only if $\pi_B P_B(x) f_B > \mu$.
Hence, $L_{A,\lambda} = L_{B,\mu}$.
\end{proof}

Second, stochastic languages are independent in some way of the initial state distribution.
\begin{proposition}\label{p-sa-norm2}
Let $A$ be a stochastic acceptor and $\lambda$ be a cut point.
Then there is a stochastic acceptor $B$, whose initial state distribution is given by a single state,
such that $L_{A,\lambda} = L_{B,\lambda}$.
\end{proposition}
\begin{proof}
Let $A=(S,\Sigma,P,\pi,f)$ be a stochastic acceptor with state set 
$S_A = \{s_1,\ldots,s_n\}$ and let $\lambda$ be a cut point.

Define the stochastic acceptor $B$ with alphabet $\Sigma$,
state set $S_B=S_A\cup\{s_0\}$, transition matrices
$$P_B(a) = \left(\begin{array}{c|c}
0       & \pi P_A(a) \\\hline
0       &            \\
\vdots  & P_A(a)     \\
0       & 
\end{array}\right),\quad a\in\Sigma,$$
and
$$\pi_B = (1,0\ldots,0),
\quad
f_B = \left\{
\begin{array}{ll}
\left(\begin{array}{c} 0\\ f_A\end{array}\right) & \mbox{if }\pi_A f_A\leq \lambda,\\
\left(\begin{array}{c} 1\\ f_A\end{array}\right) & \mbox{if }\pi_A f_A> \lambda.
\end{array}
\right.$$
Then for each $x\in\Sigma_A^*$, 
$$P_B(x) = \left(\begin{array}{c|c}
0       & \pi P_A(x) \\\hline
0       &            \\
\vdots  & P_A(x)     \\
0       & 
\end{array}\right).$$
In view of the empty word, we have $\pi_A f_A >\lambda$ if and only if $\pi_B f_B =1>\lambda$.
Thus $\epsilon\in L_{A,\lambda}$ if and only if $\epsilon \in L_{B,\lambda}$.

Moreover, for each word $x\in\Sigma^*$, we have $\pi_A P_A(x) f_A = \pi_B P_B(x) f_B$.
Thus $x\in L_{A,\lambda}$ if and only if $x\in L_{B,\lambda}$.
\end{proof}

Finally, the concept of isolated cut points is introduced.
For this, let $A$ be a stochastic acceptor.
A cut point $\lambda$ is called {\em isolated\/}\index{cut point!isolated}  for $A$
if there is a number $\Delta>0$ such that for all words $x\in\Sigma_A^*$,
\begin{eqnarray}
|\pi_AP_A(x) f_A- \lambda|\geq \Delta.
\end{eqnarray}
Isolated cut points have some meaning when conducting probabilistic experiments with stochastic acceptors.

\begin{theorem}[Rabin]
Let $A$ be a stochastic acceptor with isolated cut point~$\lambda$.
Then the language $L_{A,\lambda}$ is regular.
\end{theorem}
\begin{proof}
Let $A$ have state set $S_A=\{s_1,\ldots,s_n\}$.
Then by Prop.~\ref{p-sa-norm2},
there is a stochastic acceptor $B$ with state set $S_A = \{s_0\}\cup S_A$ such that
$L_{A,\lambda} = L_{B,\lambda}$.

Let $\equiv_L$ be the Nerode equivalence relation of the language $L=L_{A,\lambda}$.
Suppose $x$ and $y$ are words in $\Sigma^*$ such that $x\not\equiv_L y$.
Then there exists a word $z\in\Sigma^*$ such that without restriction
$xz\in L_{A,\lambda}$ and $yz\not\in L_{A,\lambda}$.
Thus $\pi_B P_B(xz) f_B> \lambda$ and $\pi_B P_B(yz) f_B \leq \lambda$.
Since $\lambda$ is isolated,
$\pi_B P_B(xz) f_B\geq \lambda +\Delta$ and $\pi_B P_B(yz) f_B \leq \lambda -\Delta$.
Therefore, $$\pi_B (P_B(x) - P_B(y))P_B(z) f_B\geq 2\Delta.$$
The column vector $P_B(z)f_B$ has nonnegative entries and the initial state vector is $\pi_B=(1,0,\ldots,0)$.
Thus the above inequality remains valid if the summation is restricted to the positive entries $p_{0j}(x)-p_{0j}(y)$, i.e.\
$$\sum_j^+ p_{0j}(x)-p_{0j}(y)\geq 2\Delta,$$
where $\sum^+$ denotes the summation over the positive terms.
Since the matrices $P_B(x)$ and $P_B(y)$ are stochastic, we obtain
$\sum_{j=0}^n p_{0j}(x) = \sum_{j=0}^n p_{0j}(y) = 1$ and thus
$\sum_{j=0}^n p_{0j}(x) - p_{0j}(y) = 0$.
Hence,
$$\sum_j^+ p_{0j}(x)-p_{0j}(y) = -\sum_j^- p_{0j}(x)-p_{0j}(y),$$
where $\sum^-$ denotes the summation over the non-positive terms.
Therefore,
$$\sum_{j=0}^n |p_{0j}(x)-p_{0j}(y)| = 2\sum_j^+ p_{0j}(x)-p_{0j}(y).$$
We may assume that $x,y$ are different from the empty word.
Then $p_{00}(x) = p_{00}(y)$ and so
$$2\Delta \leq \frac{1}{2} \sum_{j=1}^n |p_{0j}(x)-p_{0j}(y)|.$$
Thus in case of $x\not\equiv_L y$, the row vectors $\gamma_x = (p_{01}(x),\ldots,p_{0n}(x))$
and $\gamma_y = (p_{01}(y),\ldots,p_{0n}(y))$
have the $|\cdot|$ (absolute) distance of at least $4\Delta$.

However, all vectors of length $n$ with row sum~$1$ form a closed and limited (and thus compact) subset of $\RR^n$,
which can be covered by finitely many $n$-dimensional cubes of length $2\Delta$.
The above considerations show that 
each of these cubes can only contain those vectors~$\gamma_x$, which belong to the same equivalence class of the Nerode
equivalence relation.
Thus there are only finitely many classes of the Nerode equivalence.
Hence, by the Nerode-Myhill theorem, the language $L_{B,\lambda}$ is regular.
\end{proof}

\begin{example}
Consider the stochastic acceptor 
$$A = (\{s_1,s_2\},\{0,2\},\{P(0),P(2)\},\pi,f)$$ 
with
$$ 
P(0) = \begin{pmatrix} 1 & 0\\ \frac{2}{3} & \frac{1}{3}\end{pmatrix},
\quad
P(2) = \begin{pmatrix} \frac{1}{3} & \frac{2}{3}\\ 0 & 1\end{pmatrix},
\quad
\pi = \begin{pmatrix} 1 & 0\end{pmatrix},
\quad
f = \begin{pmatrix} 0 \\ 1\end{pmatrix},
$$
which is part of the 3-adic acceptor.
For each word, $x=x_1\ldots x_n\in\Sigma^*$, we have
$$\pi P(x) f = 
\frac{x_n}{3} 
+\frac{x_{n-1}}{3^2} 
\ldots
+\frac{x_1}{3^{n-1}} . $$
It follows that the topological closure of the set $\{\pi P(x) f\mid x\in\Sigma^*\}$
is precisely Cantor's discontinuum.
In particular, the set of isolated cut points of $A$ lies dense in the interval $[0,1]$. 
\EXX
\end{example}

\cleardoublepage

\addcontentsline{toc}{chapter}{Literature}

\cleardoublepage
\addcontentsline{toc}{chapter}{Index}
\printindex


\begin{thebibliography}{99}
\bibitem{buk} R.G.~Bukharaev: {\em Theorie der stochastischen Automaten}, Teubner, Stuttgart, 1995.
\bibitem{carl} J.~W.\ Carlyle: Reduced forms for stochastic sequential machines, {\it Journal Mathematical Analysis and Applications}, {\bf 7}, No.~2 (1963), 167-165.  doi: 10.1016/0022-247X(63)90045-3
\bibitem{claus} V.~Claus: {\em Stochastische Automaten}, Teubner, Stuttgart, 1971.
\bibitem{neumann} J.\ von Neumann: Probabilistic logic and the synthesis of reliable organisms from unreliable components, in: Automata Studies, C.~Shannon and J.~McCarthy (eds), {\it Annals of Mathema\-ti\-cal Studies}, {\bf 34}, Princeton Univ.\ Press, Princeton, NJ (1956).  doi: 10.1515/9781400882618-003
\bibitem{rabin} M.~O.\ Rabin: Probabilistic automata, {\it Information and Control}, {\bf 6}, No.~3 (1963), 230-245.  
doi: 10.1016/S0019-9958(63)90290-0
\bibitem{rscott} M.~O.\ Rabin, D.~Scott: Finite automata and their decision problems, {\it IBM Journal Research Development}, {\bf 3}, No.~3 (1959), 114-125.  doi: 10.1147/rd.32.0114
\bibitem{salomaa} A.\ Salomaa: {\it Theory of Automata}, Pergamon Press, Oxford, 1969.
\bibitem{shannon} C.~E.\ Shannon: A mathematical theory of communication, {\it Bell System Technical Journal}, {\bf 5}, No.~1 (1948), 379-423.  
\bibitem{starke} P.~H.\ Starke: Stochastische Ereignisse und Wortmengen, {\it Zeitschrift f\"ur Mathematische Logik und Grundlagen der Mathematik}, {\bf 12} (1966), 61-68.  doi: 10.1002/malq.19660120108
\bibitem{tura69} P.~Turakainen: Generalized automata and stochastic languages, {\em Proc.\ Amer.\ Math.\ Soc.}, 21, 303-309, 1969.
\bibitem{zim} M.N.~Cakir, K.-H.~Zimmermann: On stochastic automata over monoids, TU Hamburg, arxiv:2002.01214, 2020.
\end{thebibliography}
\end{document}